\PassOptionsToPackage{hyphens}{url} 
\PassOptionsToPackage{all}{hypcap} 
\PassOptionsToPackage{hidelinks,pdfusetitle,unicode}{hyperref}
\documentclass[fleqn]{lmcs}
\usepackage[utf8]{inputenc}
\pdfoutput=1

\usepackage{lastpage}
\lmcsdoi{18}{2}{15}
\lmcsheading{}{\pageref{LastPage}}{}{}%
{Jan.~11,~2021}{Jun.~07,~2022}{}

\usepackage{silence} 
\usepackage{resources/sensible} 
\usepackage{marvosym} 
\usepackage{adjustbox} 
\usepackage{placeins} 
\usepackage{mdframed} 
\usepackage{enumitem} 
\usepackage{nccmath} 
\usepackage{mathtools} 
\usepackage{amssymb} 
\usepackage{MnSymbol} 
\usepackage{amsthm} 
\usepackage{xfrac} 
\makeatletter
\def\@endtheorem{\endtrivlist}
\makeatother
\usepackage{bbold} 
\usepackage{ebproof} 
\usepackage[all,pdf]{xy}
\CompileMatrices
\usepackage{scalerel,stackengine}
\stackMath
\usepackage{resources/unicode-math-input}
\undef{\bibsetup}
\RequirePackage[style=american]{csquotes} 
\RequirePackage[backend=biber,style=alphabetic,minalphanames=3,maxalphanames=5,maxbibnames=9999]{biblatex}

\usepackage{resources/sensible-suffix} 

\makeatletter
\def\orcidID#1{\smash{\href{http://orcid.org/#1}{\protect\raisebox{-1.25pt}{\protect\includegraphics{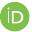}}}}}
\makeatother

\usepackage{custommacros} 


\begin{document}

\setlength\mathindent{\parindent}
\newlength\originalparindent%
\setlength\originalparindent{\parindent}
\setlist[itemize]{labelsep=*,leftmargin=1\mathindent}
\setlist[enumerate]{labelsep=*,leftmargin=1\mathindent}
\setlist[description]{leftmargin=1\mathindent}
\renewenvironment{quote}{%
  \list{}{%
    \leftmargin\parindent%
    \rightmargin\leftmargin%
  }
  \item\relax
}
{\endlist}

\title[Quotients, inductive types, and quotient inductive types]{%
    \texorpdfstring{Quotients, inductive types, \\
        \& quotient inductive types\hspace{0.21em}\phantom{,}
    }{Quotients, inductive types, and quotient inductive types}
}

\author[M.~P.~Fiore]{Marcelo P.~Fiore~\texorpdfstring{\protect\orcidID{0000-0001-8558-3492}}} 
\address{Department of Computer Science and Technology\\
University of Cambridge\\
United Kingdom}
\email{\{marcelo.fiore, andrew.pitts, s.c.steenkamp\}@cl.cam.ac.uk}

\author[A.~M.~Pitts]{Andrew M.~Pitts~\texorpdfstring{\protect\orcidID{0000-0001-7775-3471}}} 

\author[S.~C.~Steenkamp]{S.~C.~Steenkamp\texorpdfstring{\textsuperscript{(\protect\Letter)}~\protect\orcidID{0000-0003-3105-4098}}{}} 



\begin{abstract}
  This paper introduces an expressive class of indexed
  quotient-inductive types, called QWI types, within the framework of
  constructive type theory. They are initial algebras for indexed
  families of equational theories with possibly infinitary operators
  and equations. We prove that QWI types can be derived from quotient
  types and inductive types in the type theory of toposes with natural
  number object and universes, provided those universes satisfy the
  Weakly Initial Set of Covers (WISC) axiom. We do so by constructing
  QWI types as colimits of a family of approximations to them defined
  by well-founded recursion over a suitable notion of size, whose
  definition involves the WISC axiom. We developed the proof and
  checked it using the Agda theorem prover.
\end{abstract}

\maketitle

\section{Introduction}%
\label{sec:introduction}

Inductive types are an essential feature of many type theories and the
proof assistants and programming languages that are based upon
them. Homotopy Type Theory \parencite{HoTT} introduces a powerful
extension to the notion of inductive type: higher inductive
types~(HITs). To define an ordinary inductive type one declares the
ways in which its elements are constructed, and these constructions
may depend inductively on smaller elements. To define a HIT one not
only declares element constructors, but also declares equality
constructors valued in identity types (possibly iterated ones),
specifying how the element constructors~(and possibly the equality
constructors) are related. In the language of homotopy theory, these
correspond to points, paths, surfaces, and so on.

In this paper, rather than using HoTT, we work in the simpler setting
of Extensional Type Theory~\parencite{martin-lof1982constructive},
where any propositional equality is reflected as a definitional
equality. However, our results also hold for an intensional type
theory with the uniqueness of identity proofs axiom (UIP), as
demonstrated by our Agda development which can be found
at~\parencite{fiore2021agdacode}.  In either case identity types are
trivial in dimensions higher than one, so that any two proofs of
identity are equal. Nevertheless, as \textcite{altenkirch2016type}
point out, HITs are still useful in such a one-dimensional
setting. They introduced the term \emph{quotient inductive type} (QIT)
for this truncated form of HIT\@.

One specific advantage of QITs over inductive types is that they
provide a natural way to study universal algebra in type theory, since
general algebraic theories cannot be represented in an
equation-free way just using inductive notions
\parencite{malbos2016homological}.  As an example, the free
commutative monoid on carrier \(X\) is given by the quotient inductive
type of finite multisets\footnotemark{}, \(\Bag\,X\), with
constructors:
\begin{equation}\label{bag_qit}
    \begin{aligned}
        {} [] &: \Bag\,X
        \\ \_{∷}\_ &: X → \Bag\,X → \Bag\,X
        \\ \swap &: ∏_{x,y : X} ∏_{\zs : \Bag\,X} x ∷ y ∷ \zs = y ∷ x ∷ \zs
    \end{aligned}
\end{equation}
This can be seen as the element constructors for lists ($[]$ and $∷$),
along with an equality constructor \swap{} that transposes two
elements and thereby allows all reorderings
of a list to be identified. Notice that the \swap{} constructor lands
in the identity type on \(\Bag\,X\) (just written as $\_{=}\_$),
making it an \emph{equation} constructor, compared to the
\emph{element} constructors, \([]\) and \(\_{∷}\_\), which are familiar
from ordinary inductive type definitions. There is no need to include
a constructor for the congruence property of $\_{∷}\_$ since that is
automatic from the congruence properties of identity types.

\footnotetext{Other presentations of multisets are available, aside
  from the two presentations given here. Incidentally, the algebra of
  monoids has an equation-free
  presentation~\parencite[section~23]{KellyGM:unittc},
  \parencite{KellyGM:oncdt} – the algebra of Lists – and this is why
  no equation for associativity is needed in either definition of
  multisets.}

An alternative presentation of multisets, whose equations may be
easier to work with, is:%
\begin{equation}%
    \label{alt-bag_qit}
    \begin{aligned}
        {} [] &: \Bag'\,X
        \\ \_{∷}\_ &: X → \Bag'\,X → \Bag'\,X
        \\ \comm &: ∏_{x,y : X} ∏_{\as,\bs,\cs : \Bag'\,X}
            \as = y ∷ \cs → x ∷ \cs = \bs → x ∷ \as = y ∷ \bs
    \end{aligned}
\end{equation}
Notice that the \comm{} equality constructor is \emph{conditional} on
two equations, \(\as = y ∷ \cs\) and \(x ∷ \cs = \bs\), in order to
deduce the final equation \(x ∷ \as = y ∷ \bs\). The QWI-types
introduced in this paper are not able to encode such conditional QITs
(at least, not obviously so); see the discussion in the Conclusion.

Given the parameter $X$, the types $\Bag\,X$ and $\Bag'\,X$ are single
quotient-inductive types, rather than indexed families defined
quotient-inductively. As an example of the latter, consider the go-to
example of an inductively defined family, the vector type
(length-indexed lists). This has a commutative variant,
\(\CommVec\,X\,n\) for \(n : ℕ\), which makes it an example of an
indexed QIT\@. It has constructors:
\begin{equation}%
  \label{abvec_qit}
  \begin{aligned}
    {} [] &: \CommVec\,X\,0
    \\ \_{∷}\_ &: X → ∏_{i : ℕ}\CommVec\,X\,i → \CommVec\,X\,(i + 1)
    \\ \swap &: ∏_{x,y : X} ∏_{i : ℕ} ∏_{\zs : \CommVec\,X\,i}
    x ∷ y ∷ \zs = y ∷ x ∷ \zs
  \end{aligned}
\end{equation}
Our final example of a QIT is unordered countably-branching trees
over \(X\), called \(\infTree\,X\), with constructors:
\begin{equation}%
  \label{infTree_qit}
  \begin{aligned}
    {} \leaf &: \infTree\,X
    \\ \node &: X → (ℕ → \infTree\,X) → \infTree\,X
    \\ \perm &: ∏_{x : X} ∏_{b : ℕ → ℕ} ∏_{b' : \isIso\,b}
    ∏_{f : ℕ → \infTree\,X} \node\,x\,f = \node\,x\,(f ∘ b)
  \end{aligned}
\end{equation}
where elements of \(\isIso\,b\) witness that \(b\) is a bijection. (As
with \CommVec{}, one could similarly consider a depth-indexed variant
of \(\infTree\).) The significance of this example is that the
$\node\,x\,\_$ and $\perm\,x\,b\,b'\,\_$ constructors have infinite
arity, making this an example of an \emph{infinitary} QIT\@.  This type
is one of the original motivations for considering QITs
\parencite{altenkirch2016type}, since they can enable constructive
versions of structures that classically use non-constructive choice
principles, as we explain next.

The examples of QITs in~\eqref{bag_qit}--\eqref{abvec_qit} only
involve element constructors of \emph{finite} arity. For example in~\eqref{bag_qit}, $[]$ is nullary and for each $x,y:X$, $x∷\_$ and
$\swap\,x\,y\,\_$ are unary. Consequently $\Bag\,X$ is isomorphic to
the type obtained from the ordinary inductive type of finite lists
over $X$ by quotienting by the congruence generated by $\swap$.  By
contrast,~\eqref{infTree_qit} involves element and equality
constructors with countably infinite arity. So if one first forms the
ordinary inductive type of \emph{ordered} (or planar)
countably-branching trees (by dropping the equality constructor
\perm{} from the declaration) and then quotients by a suitable
relation to get the equalities specified by \perm, one appears to need
the axiom of countable choice to be able to lift the \node{} element
constructor to the quotient
\parencite[Section~2.2]{altenkirch2016type}.

The construction of the Cauchy reals as a higher inductive-inductive
type~\parencite[Section~11.3]{HoTT} provides a similar, but more
complicated example where use of countable choice is avoided. So it
seems that without some form of choice principle, quotient inductive
types are strictly more expressive than the combination of ordinary
inductive types with quotient types. Indeed,
\textcite[Section~9]{lumsdaine2019semantics} turn an infinitary
equational theory due to \textcite[Section~9]{BlassA:worfac} into a
higher-inductive type that cannot be proved to exist in ZF set theory
without the Axiom of Choice. In this paper we show that in fact a much
weaker and constructively acceptable form of choice is sufficient for
constructing indexed quotient inductive types. This is the
\emph{Weakly Initial Set of Covers} (WISC) axiom, originally called
the ``Type Theoretic Collection Axiom'' \textsc{TTCA}$_f$ by
\textcite{StreicherT:reamc} and (in the context of constructive set
theory) the ``Axiom of Multiple Choice''
by \textcite{vandenberg2014axiom}.  WISC is constructively acceptable
in that it is known to hold in the internal logic of a wide range of
toposes~\parencite{vandenberg2014axiom} (for example, all Grothendieck
and all realizability toposes over the topos of classical sets),
including those commonly used in the semantics of Type
Theory. \Cref{sec:wisc} reviews WISC and our motivation for using it
herein.

We make two contributions:
\begin{itemize}
\item First we define a class of indexed quotient inductive types
  called \emph{QWI-types} and give elimination and computation rules
  for them (\Cref{sec:qw-types}). The usual W-types of
  \textcite{martin-lof1982constructive} are inductive types giving the
  algebraic terms over a possibly infinitary signature. One specifies
  a QWI-type by giving a family of equations between such terms. So
  such types give initial algebras for (families of) possibly
  infinitary algebraic theories. They can encode a very wide range of
  examples of, possibly infinitary, indexed quotient inductive types
  (\Cref{sec:qit_encoding}). In classical set theory with the Axiom of
  Choice, QWI-types can be constructed simply as quotients of the
  underlying Indexed W-type---hence the name.

\item Second, our main result (\Cref{thm:main}) is that with only
  WISC it is still possible to construct QWI-types from inductive
  types and quotients in the constructive type theory of toposes, but
  not simply by quotienting a W-type.  Instead, quotienting is
  interleaved with an inductive construction. To make sense of this
  and prove that the resulting type has the required universal
  property, we construct the type as the colimit of a family of
  approximations to it defined by well-founded recursion over a
  suitable notion of \emph{size} (a constructive version of what
  classically would be accomplished with a sequence of transfinite
  ordinal length). Our notion of size is developed in
  \Cref{sec:size}. Sizes are elements of a W-type equipped with the
  \emph{plump}~\parencite{TaylorP:intso} well-founded order; and WISC
  allows us to make the W-type big enough that a given polynominal
  endofunctor preserves size-indexed colimits
  (\Cref{cor:poly-cocontinuity}). This fact is then applied in
  \Cref{sec:construction} to form QWI-types as size-indexed colimits.
\end{itemize}

\subsection*{Note}
This paper is a greatly revised and expanded version of our earlier
paper \parencite{fiore2020constructing}, which introduced QW-types
(see~\Cref{def:qwt} which is the non-indexed version
of~\Cref{def:qwit}). There we made use of \emph{sized
types}~\parencite{AbelA:typbti} as they are implemented in the Agda
theorem prover~\parencite{Agda261}. Though the results of that paper
were formalised in Agda version 2.5.2, unfortunately, in the current
version of Agda, 2.6.2, sized types are logically unsound. In this
paper, not only do we consider the more expressive indexed form of
QW-types, namely QWI-types, we also put the construction of QW-types
on a firm semantic footing by avoiding Agda-style sized types (and
positivity assumptions on quotient types,
see~\Cref{motivatingtheconstruction}) in favour of a well-founded
notion of size built from WISC\@. The Agda formalization of the results
in this paper is available at~\parencite{fiore2021agdacode}.

\FloatBarrier
\section{Type theory}%
\label{sec:type-theory}

The results in this paper are proved in a version of Extensional Type
Theory~\parencite{martin-lof1982constructive} that is an internal
language for toposes~\parencite{JohnstonePT:skeett} with natural
numbers object and universes~\parencite{StreicherT:unit}. We use this
type theory in an informal way, in the style and notation of the HoTT
Book~\parencite{HoTT}. Thus dependent function types are written
$∏_{x:A}B\,x$ and non-dependent ones as either $A→ B$ or $B^A$.
Dependent product types are written $∑_{x:A}B\,x$ and non-dependent
ones as $A × B$. Coproduct types are written $A+B$ (with injections
$ι₁:A → A+B$ and $ι₂:B → A+B$); and finite types are written $𝟘$ (with
no elements), $𝟙$ (with a single element $0$), $𝟚$ (with two elements
$0,1$), etc. The identity type for $x:A$ and $y:A$ is $x =_A y$, or
just $x = y$ when the type $A$ is known from the context. When we need
to refer to the judgement that $x$ and $y$ are definitionally equal,
we write $x ≡ y$. Unlike in the HoTT Book~\parencite{HoTT}, the type
$x = y$ is an extensional identity type and so it is inhabited iff $x≡
y$ holds. In particular, proofs of identity are unique when they exist
and so it makes sense to use McBride's heterogeneous form of identity
type~\parencite{mcbride1999thesis} wherever possible.  Given $x:A$ and
$y:B$, we denote the heterogeneous identity type by $x\heq y$; thus
this type is inhabited iff $A ≡ B$ and $x ≡ y$.

The type theory has universes of types, $𝓤_0 : 𝓤_1 : 𝓤_2,\ldots$ which
we write just as $𝓤$ when the universe level is immaterial. We use
Russell-style universes (there is no syntactic distinction between an
element $A:𝓤$ and the corresponding type of its elements). Influenced
by Agda (see below) and unlike the HoTT Book~\parencite{HoTT}, we do
not assume universes are cumulative, but instead use Agda-style
closure under $\prod$-types: if $A:𝓤_i$ and $B:A→𝓤_j$, then
$∏_{x:A} B\,x : 𝓤_{\max i j}$ (and similarly for $\sum$-types). Since
we only consider toposes that have a natural numbers object, we can
assume the universes are closed under forming not just
W-types~\parencite[Proposition~3.6]{MoerdijkI:weltc}, but also all
inductively defined indexed families of
types~\parencite{GambinoN:welftd}. Indexed W-types are considered in
more detail in the next section.

The lowest universe contains an impredicative universe $\Prop$ of
propositions, corresponding to the subobject classifier in a
topos. $\Prop$ contains the identity types, is closed under
intuitionistic connectives ($∧$, $∨$, $→$, $↔$, etc.) and quantifiers
($∀(x:A).φ(x)$ and $∃(x:A).φ(x)$, for $A$ in any universe), and satisfies
propositional extensionality:
\begin{equation}%
  \label{eq:propext}
   \propext : ∀(p,q:\Prop).\, (p ↔ q) → {p = q}
 \end{equation}
 Being an extensional type theory, we also have function
 extensionality:
\begin{equation}%
  \label{eq:funext}
  \funext : ∀\left(f,g :∏_{x:A}B\,x\right).\, (∀(x:A).\, f\,x = g\, x)
  → f = g
\end{equation}
Note that like the universes $𝓤_i$, $\Prop$ is also a Russell-style
universe, in that we do not make a notational distinction between a
proposition $p:\Prop$ and the type of its proofs. Given $A:𝓤$ and
$φ:A→ \Prop$, we regard the comprehension type $\{x:A ∣ φ\,x\}$ in $𝓤$
as synonymous with the dependent product $∑_{x:A}\,φ\,x$.

Toposes have coequalizers and effective
epimorphisms~\parencite[A2.4]{JohnstonePT:skeett}. Correspondingly the
type theory contains quotient types, with notation as in
\Cref{fig:quot} (recall that $\heq$ stands for heterogeneous
identity).  These can be constructed in the usual way via equivalence
classes, using $\Prop$'s impredicative quantifiers and the fact that
toposes satisfy unique choice:~
\begin{equation}%
  \label{eq:uniquechoice}
  \uniquechoice : (∃(x:A).∀(y:A).x = y) → A
\end{equation}
Thus $\uniquechoice$ is a function mapping proofs that $A$ has exactly
one element to a name for that one element. From $\uniquechoice$ it
follows that given $A:𝓤$, $B:𝓤^A$ and $φ:∏_{x:A}(B\,x → \Prop)$, from
a proof of $∀(x:A)∃!(y:B\,x).\;φ\,x\,y$ we get there is a unique
$f:∏_{x:A} B\,x$ such that $∀(x:A).\; φ\,x\,(f\,x)$.

\begin{figure}
  \begin{mdframed}
    \begin{itemize}
    \item  Given $A:𝓤$ and $R:A → A → \Prop$, we have:
      \begin{align*}
        &A/R : 𝓤\\
        &[\_]_R: A → A/R\\
        &\mathsf{qeq}_R : ∀(x,y:A)·\,R\,x\,y → [x]_R = [y]_R
      \end{align*}
    \item Furthermore, given $B:A/R → 𝓤$, $f: ∏_{x:A} B([x]_R)$, and
        \(p : ∀(x,y:A)·\, R\,x\,y → f\,x \heq f\,y\), we have:
      \begin{align*}
        &\mathsf{qelim}_R\,B\,f\,p : ∏_{z:A/R} B\,z\\
        &\mathsf{qcomp}_R\,B\,f\,p :
          ∀(x:A)·\,\mathsf{qelim}\,B\,f\,[x]_R = f\,x
      \end{align*}

    \item Quotients of equivalence relations are effective: writing
      $\mathsf{ER}_R$ for the proposition that $R$ is reflexive,
      symmetric and transitive, we have:
      \begin{align*}
        &\mathsf{qeff}_R : \mathsf{ER}_R → ∀(x,y:A)·\, [x]_R = [y]_R  →
          R\,x\,y
      \end{align*}
    \end{itemize}
  \end{mdframed}
  \caption{Quotient types}%
  \label{fig:quot}
\end{figure}

\subsection*{Agda development}

We have developed and checked a version of the results in this paper
using the Agda theorem prover~\parencite{Agda261}. In particular our Agda
development gives the full details of the construction of the indexed
version of QW-types, whereas in the paper, for readability, we only give
the construction in the non-indexed case (\cref{sec:construction}).

Being intensional and predicative, the type theory provided by Agda is
weaker than the one described above, but can be soundly interpreted in
it. One has to work harder to establish the results with Agda, since
two expressions that are definitionally equal in Extensional Type
Theory may not be so in Agda and hence one has to produce a term in a
corresponding identity type to prove them equal. On the other hand,
when a candidate term is given, Agda can decide whether or not it is a
correct proof, because validity of the judgements of the type theory
it implements is decidable\footnote{though not always feasibly so}, in
contrast with the situation for Extensional Type
Theory~\parencite{HofmannM:extcit}. Our development is also made
easier by extensive use of the predicative universes of
proof-irrelevant propositions that feature in recent versions of
Agda. Not only are any two proofs of such propositions definitionally
equal, but inductively defined propositions (such as the well-founded
ordering on sizes used in \Cref{sec:construction}) can be eliminated
in proofs using dependent pattern matching, which is a very great
convenience.  We still need these propositions to satisfy the
extensionality~\eqref{eq:propext} and unique
choice~\eqref{eq:uniquechoice} properties, so we add them as
postulates in Agda. Also, the impredicative construction of quotient
types is not available, so we get quotients as in \Cref{fig:quot} by
postulating them and using a user-declared rewrite to make their
computation rule a definitional equality. Although function
extensionality~\eqref{eq:funext} is not provable in Agda's core type
theory, it becomes so once one has such quotient types. Our Agda
development is available at~\parencite{fiore2021agdacode}.


\FloatBarrier
\section{Indexed polynomial functors and equational systems}%
\label{sec:qw-types}

In this section we introduce a class of indexed quotient inductive
types, called \emph{QWI-types}, which are indexed families of free
algebras for possibly infinitary equational theories. To begin with
we describe the simpler, non-indexed case (\emph{QW-types}) and then
treat the general, indexed case in \cref{sec:qwi-types}.

\subsection{W-types}%
\label{sec:w-types}

We recall some facts about types of well-founded trees,
the W-types of \textcite{martin-lof1982constructive}. We take
\emph{signatures} (also known as
\emph{containers}~\parencite{abbott2005containers}) to be elements of
the dependent product
\begin{equation}\textstyle%
  \label{eq:qw-types-2}
  \Sig ≝ ∑_{A:𝓤}(A → 𝓤)
\end{equation}
So a signature is given by a pair \(Σ=(A,B)\) consisting of a type
\(A:𝓤\) and a family of types \(B:A → 𝓤\). Each such signature
determines a polynomial endofunctor
\parencite{GambinoN:welftd,abbott2005containers}
\(S_Σ:𝓤→𝓤\) whose value at
$X:𝓤$ is the following dependent product
\begin{equation}\textstyle%
  \label{def_S_endofunctor}
  S_Σ(X) ≝ ∑_{a:A}(B\,a → X)
\end{equation}
and whose action on a function $f:X→Y$ in $𝓤$ is the function
$S_Σ f:S_Σ(X)→ S_Σ(Y)$ where
\begin{equation}%
  \label{def_S_endofunctor-action}
  S_Σ f\,(a,b) ≝ (a, f∘ b) \qquad (a:A, b:B\,a → X)
\end{equation}
An \emph{\(S_Σ\)-algebra} is by definition an element of the
dependent product
\begin{equation}\textstyle%
  \label{eq:qw-types-7}
  \Alg_Σ ≝ ∑_{X:𝓤}(S_Σ(X) → X)
\end{equation}
\(S_Σ\)-algebra morphisms \((X,α)→(X',α')\) are given by functions
\(h:X→ X'\) together with an element of the type
\begin{equation}%
  \label{eq:qw-types-8}
  \isHom_{α,α'}\,h ≝ ∏_{a:A}∏_{b:B\,a → X} (α'(a, h ∘ b) = h(α(a, b)))
\end{equation}
Then the W-type \(🅆_Σ\) determined by a signature \(Σ\) is the
underlying type of an initial object in the evident category of
\(S_Σ\)-algebras. More generally, \textcite{dybjer1997representing}
shows that the initial algebra of any non-nested, strictly positive
endofunctor on $𝓤$ is given by a W-type; and
\textcite{abbott2005containers} extend this to the case with nested
uses of W-types as part of their work on containers. (These proofs
take place in extensional type
theory~\parencite{martin-lof1982constructive}, but work just as well
in intensional type theory with uniqueness of identity proofs and
function extensionality, which is what we use in Agda.)

More concretely, given a signature $Σ=(A,B)$, if one thinks of
elements $a:A$ as names of operation symbols whose (not necessarily
finite) arity is given by the type $B\,a:𝓤$, then the elements of the
W-type \(🅆_Σ\) represent the closed algebraic terms~(i.e.~well-founded
trees) over the signature. From this point of view it is natural to
consider not only closed terms solely built up from operations, but
also open terms additionally built up with variables drawn from some
type $V$. As well as allowing operators of possibly infinite arity, we
also allow terms involving possibly infinitely many variables (example~\eqref{infTree_qit} involves such terms).  Categorically, the type
$T_Σ(V)$ of such open terms is the free $S_Σ$-algebra on $V$ and is
another W-type, for the signature obtained from $Σ$ by adding the
elements of $V$ as nullary operations. Nevertheless, it is convenient
to give a direct inductive definition: the value of \(T_Σ : 𝓤 → 𝓤\) at
some some \(V : 𝓤\), is the inductively defined type with constructors
\begin{equation}\label{eq:Tconstructors}
  \begin{aligned}
    η &: V → T_Σ\,V\\
    σ &: S_Σ (T_Σ\,V) → T_Σ\,V
  \end{aligned}
\end{equation}
Given an \(S_Σ\)-algebra \((X, α) : \Alg_Σ\) and a function \(f : V →
X\), the unique morphism of \(S_Σ\)-algebras from the free
\(S_Σ\)-algebra on \(V\) to \((X, α)\) that extends \(f\), \(\_ ⋙ f :
(T_Σ(V),σ) → (X,α)\), has underlying function \(T_Σ(V) → X\) mapping
each \(t : T_Σ(V)\) to the element \(t ⋙ f\) defined as follows by
recursion on the structure of \(t\):
\begin{equation}\label{def_bind}
    \begin{aligned}
        η\,x ⋙ f &≝ f\,x \\
        σ\,(a, b) ⋙ f &≝ α\,(a, λ\,x · b\,x ⋙ f)
    \end{aligned}
\end{equation}
As the notation suggests, \(⋙\) is the Kleisli lifting operation
(“bind”) for a monad structure on \(T_Σ\). Indeed, \(T_Σ\) is the free
monad on the endofunctor \(S_Σ\) and in particular the functorial
action of \(T_Σ\) on $f:V→V'$ yields the function
\(T_Σ f : T_Σ(V)→ T_Σ(V')\) given by
\begin{equation}%
  \label{def_Taction}
  T_Σ f\, t ≝ t  ⋙ (η ∘ f) \qquad (t : T_Σ(V))
\end{equation}

\subsection{QW-types}%
\label{sec:qwt}

The notion of \emph{QW-type} that we introduce in this section is
obtained from that of a W-type by considering not only the algebraic
terms over a given structure, but also equations between terms. To
code equations we use a type-theoretic rendering of a
categorical notion of equational system introduced by
\citeauthor{fiore2009construction}, referred to as \emph{term
  equational system}~\parencite[Section~2]{fiore2009construction} and
as \emph{monadic equational
  system}~\parencite[Section~5]{fiore2013equational}, here
instantiated to the special case of free monads on signatures.

\begin{definition}%
  \label{def:syse}
  A \emph{system of equations} over a signature \(Σ : \Sig\) is specified by:
  \begin{itemize}
  \item A type \(E : 𝓤\), whose elements \(e :
    E\) can be thought of as names for each equation.
  \item  A function \(V : E → 𝓤\). For each equation \(e:E\), the type
    \(V\,e : 𝓤\) represents the variables used in the
    equation \(e\).
  \item For each equation \(e:E\), two elements called \(l\,e\) and
    \(r\,e\) of type \(T_Σ(V\,e)\), which is the  free \(S_Σ\)-algebra on
    the variables \(V\,e\). So \(l\,e\) and \(r\,e\) can be thought
    of as the abstract syntax trees of terms with some leaves being
    free variables drawn from \(V\,e\).
  \end{itemize}
  So systems of equations are elements of the dependent product
  \begin{equation}
    \Syseq_Σ ≝ ∑_{E : 𝓤} ∑_{V : E → 𝓤}
    \bigg(∏_{e : E} T_Σ(V\,e)\bigg)
    × \bigg(∏_{e : E} T_Σ(V\,e)\bigg)
  \end{equation}
  An \(S_Σ\)-algebra \(α : S_Σ(X) → X\) \emph{satisfies} the system of equations
  \(ε ≡ (E,V,l,r) : \Syseq_Σ\) if there is a proof of
  \begin{equation}\label{def_Sat}
    \Sat_{α,ε}(X) ≝ ∀(e : E).∀(ρ : V\,e → X).\; (l\,e ⋙ ρ) = (r\,e ⋙ ρ)
  \end{equation}
\end{definition}
The category-theoretic view of QW-types is that they are simply
\(S_Σ\)-algebras that are initial among those satisfying a given
system of equations:

\begin{definition}[\textbf{QW-type}]%
  \label{def:qwt}
  A \emph{QW-type} for a signature \(Σ ≡ (A, B) : \Sig\) and a system of
  equations \(ε ≡ (E, V, l, r) : \Syseq_Σ\) is given by a type \(\QW :
  𝓤\) equipped with an \(S_Σ\)-algebra structure and a proof that it satisfies
  the equations. Thus there are functions
  \begin{align}
    \qwintro &: S_Σ(\QW) → \QW\label{def_qwintro}\\
    \qwequate &: \Sat_{\qwintro,ε}(\QW)\label{def_qwequate}
                \intertext{%
                together with functions that witness that it is an initial such
                algebra:
                }
                \qwrec &: ∏_{X : 𝓤} ∏_{α : S_Σ\,X → X} \Sat_{α,ε}(X)
                         → \QW → X\label{def_qwrec}\\
    \qwrechom &: ∏_{X : 𝓤}∏_{α : S_Σ\,X → X}
                ∏_{p : \Sat_{α,ε}(X)} \isHom\,(\qwrec\,X\,α\,p)\label{def_qwrechom} \\
    \qwuniq &:
              \prod_{X : 𝓤}∏_{α : S_Σ\,X → X}∏_{p : \Sat_{α,ε}(X)}∏_{f : \QW → X}
              \isHom\,f → \qwrec\,X\,α\,p = f\label{def_qwuniq}
  \end{align}
\end{definition}

\begin{remark}[\textbf{Free algebras}]%
  \label{remarkFreeAlgebras}
  \Cref{def:qwt} defines QW-types as initial algebras for
  systems of equations $ε$ over signatures $Σ$.  More generally, the
  \emph{free} $(Σ,ε)$-algebra on a type $X:𝓤$ is a type
  $F_{Σ,ε}\,X : 𝓤$ equipped with an $S_Σ$-algebra structure
  $S_Σ(F_{Σ,ε}\,X) → F_{Σ,ε}\,X$ satisfying the system of equations
  $ε$ and an inclusion of generators $η_X:X → F_{Σ,ε}\,X$ which is
  universal among such data. Initial algebras are the $X=𝟘$ special
  case of free algebras. But once one has them, one also has free
  algebras, by change of signature: $F_{Σ,ε}\,X$ is the QW-type for
  the signature $Σ_X$ and system of equations $ε_X$ defined as
  follows. If $Σ=(A,B)$, then $Σ_X=(X+A,B_X)$ where $B_X: X+A → 𝓤$
  satisfies $B_X(ι₁\,x) = 𝟘$ and $B_X(ι₂\,a) = B\,a$. And if
  $ε=(E,V,l,r)$, then $ε_X=(E,V,l_X,r_X)$ where for each $e:E$,
  $l_X\,e = (l\,e ⋙ η)$ and $r_X\,e = (r\,e ⋙ η)$ (using the
  $S_Σ$-algebra structure $s$ on $T_Σ(V\,e)$ given by
  $s(a,b) = σ(ι₂\,a,b)$).
\end{remark}
The definitions of \(S_Σ\) in~\eqref{def_S_endofunctor} and
\(\Sat_{α,ε}\) in~\eqref{def_Sat} suggest that QW-types are an
instance of the notion of quotient-inductive type
\parencite{altenkirch2016type}, with \qwintro{} as the element
constructor and \qwequate{} as the equality constructor. To show that
QW-types are indeed quotient-inductive types, they need
to have the requisite dependently-typed elimination and
computation properties\footnotemark{} for these elements and equality
constructors. We show that these follow
from~\eqref{def_qwrec}–\eqref{def_qwuniq} and function extensionality. To state
this proposition we need a dependent version of the bind operation~\eqref{def_bind}. %
\footnotetext{In this paper we work with extensional type theory, so
  the computation property,~\eqref{qwiElimComp}, is also a
  definitional equality. In our Agda development we work in
  intensional type theory, so there we only establish the
  computation property up to propositional equality; using the
  terminology of \textcite{shulman2018brouwer}, those are \emph{typal}
  indexed quotient-inductive types.
}
For each ``motive'' \(P\) and induction step \(p\),
\begin{equation}
    \begin{aligned}\label{elim_motive_and_step}
        P &: \QW → 𝓤 \\
        p &:  ∏_{a : A} ∏_{b : B\,a → \QW}
        \bigg(\prod_{x : B\,a} P(b\,x)\bigg) → P(\qwintro(a, b))
    \end{aligned}
\end{equation}
as well as a type \(X : 𝓤\), a substitution \(ρ : \left(X → ∑_{x
: \QW} P\, x\right)\), and term \(t : T_Σ\,X\), we have an
element \(\lift_{P,X}\,p\,ρ\,t : P\,(t ⋙( π₁ ∘ ρ))\) defined by recursion on
the structure of \(t\):
\begin{equation}%
  \label{eq:lift}
    \begin{array}{lcl}
        \lift_{P,X}\,p\,ρ\,(η\,x) &≝& π₂\,(ρ\,x) \\
        \lift_{P,X}\,p\,ρ\,(σ\,(a, b)) &≝&
            p\,a\,(λ\,x · b\,x ⋙ (π₁ ∘ ρ))\,((\lift_{P,X}\,p\,ρ) ∘ b)
    \end{array}
\end{equation}
Note that the substitution \(ρ\) gives terms in the “dependent telescope”
\(∑_{x:\QW} P\,x\), which will be written as \(P'\).

\begin{proposition}
  For a QWI-type as defined above, given \(P\) and \(p\) as in~\eqref{elim_motive_and_step}, and a term \(p_\resp\) of type
  \begin{equation}\label{ind_step_respects_equations}
    ∏_{e : E}
    ∏_{ρ : V\,e → ∑_{x : \QW} P\,x}
    \lift_{P,X}\,p\,ρ\,(l\,e) == \lift_{P,X}\,p\,ρ\,(r\,e)
  \end{equation}
  there are elimination and computation terms
  \begin{equation}\label{qwiElimComp}
    \begin{aligned}
      \qwelim &: ∏_{x : \QW} P\,x \\
      \qwcomp &: ∏_{a : A} ∏_{b : B\,a → \QW}
      \qwelim\,(\qwintro\,(a,b)) = p\,a\,b\,(\qwelim ∘ b)
    \end{aligned}
  \end{equation}
\end{proposition}
(Note that~\eqref{ind_step_respects_equations} uses heterogeneous
identity $\heq$, because \(\lift_{P,X}\,p\,ρ\,(l\,e)\) and
\(\lift_{P,X}\,p\,ρ\,(r\,e)\) inhabit different types, namely
\(P\,(l\,e ⋙ (π₁ ∘ ρ))\) and \(P\,(r\,e ⋙ (π₁ ∘ ρ))\) respectively.)
\begin{proof}
  To define the eliminator, we must first use the algebra map on the more
  general type \(P' ≝ ∑_{x : \QW} P\,x\). The
  algebra map, \qwrec, requires that this type \(P'\) has an algebra structure
  \(β : S_Σ(P') → P'\), which is given by:
  \begin{equation}
    β\,(a, b) ≝ (\qwintro (a, π₁ ∘ b), p\,a\,(π₁ ∘ b)\,(π₂ ∘ b))
  \end{equation}
  Moreover, we must show that the recursor satisfies the equations
  by giving a proof \(s: \Sat_{β,ε}(P')\), which we do pointwise on
  the two elements of the dependent product \(P'\). Taking
  projections distributes (possibly dependently) over $ ⋙$; so to
  construct $s$ it suffices that, given any \(e : E\) and
  \(ρ : V\,e → P'\), we have the two terms:
  \begin{equation}\label{beta_alg_respectful}
    \begin{array}{rcrcl}
      \qwequate\,e\,ρ &:& l\,e ⋙ (π₁ ∘ ρ) &=& r\,e ⋙ (π₁ ∘ ρ) \\
      p_\resp\,e\,ρ &:& \lift\,p\,ρ\,(l\,e) &==& \lift\,p\,ρ\,(r\,e)
    \end{array}
  \end{equation}
  Then the eliminator is defined by taking the second projection of this
  recursor.
  \begin{equation}
    \qwelim ≝ π₂ ∘ \qwrec\,P'\,β\,s
  \end{equation}
  Given an element \((a,b) : S_Σ(\QW)\), we can prove that applying
  the eliminator to it computes as expected, that is
  \(\qwelim\,(\qwintro\,(a,b)) = p\,a\,b\,(\qwelim ∘ b)\), as follows,
  where we write $r'$ for $\qwrec\,P'\,β\,s$:
  \begin{align*}
    &\qwelim\,(\qwintro\,(a,b))
    \\ =={}& π₂\,(r'(\qwintro\,(a,b)))
           &\text{def.\ of \qwelim}
    \\ =={}& π₂\,(β\,(a,r' ∘ b))
           &\text{using \qwrechom{} against \(\qwintro\,(a,b)\)}
    \\ =={}& p\,a\,(π₁ ∘ r' ∘ b)\,(π₂ ∘ r' ∘ b)
           &\text{def.\ of \(β\)}
    \\ =={}& p\,a\,b\,(π₂ ∘ r' ∘ b)
           &\parbox[c]{0.43\textwidth}{\raggedleft{}\(r'\) preserves \QW{} in the
           first component; follows from \qwrechom{} and \qwuniq{}}
    \\ =={}& p\,a\,b\,(\qwelim ∘ b)
           &\text{def.\ of \qwelim}
           &\qedhere
  \end{align*}
\end{proof}

\subsection{QWI-types}%
\label{sec:qwi-types}

In this section we consider \emph{indexed} quotient inductive types
specified by families of (possibly infinitary) equational theories. As
for ordinary W-types, the indexed version of QW-types gives convenient
expressive power; see~\Cref{sec:qit_encoding}. To describe the indexed
case we first introduce some notation for indexed families of types
and functions.

Given a type \(I:𝓤\), when we map between two \(I\)-indexed types, say
\(A : 𝓤ᴵ\) to \(B : 𝓤ᴵ\), we will write
\begin{equation}
  A ⇁ B ≝ ∏_{k : I} Aₖ → Bₖ
\end{equation}
for the function type. The composition of \(f: A ⇁ B\) and
\(g: B ⇁ C\) will just be written as
\begin{equation}
  g ∘ f ≝ λ i.λ x.\,gᵢ(fᵢ\,x)
\end{equation}
We also often need to define an indexed family of types \(B_i\,a : 𝓤ᴵ\)
which is dependent on some $i:I$ and element \(a : Aᵢ\) for \(A : 𝓤ᴵ\). The type
of such families will be written as
\begin{equation}
  A ⥗ 𝓤ᴵ ≝  ∏_{i : I} (Aᵢ → 𝓤ᴵ)
\end{equation}
We take our notion of an indexed signature for a WI-type (indexed
W-type) to be a particular case of indexed containers
\parencite{abbott2005containers}, namely an element of the type
\(∑_{I : 𝓤} ∑_{A : 𝓤ᴵ} ∏_{i : I} (Aᵢ → 𝓤ᴵ)\), which with the above
notational conventions we write as
\begin{equation}
    \Sig ≝ ∑_{I : 𝓤} ∑_{A : 𝓤ᴵ} (A ⥗ 𝓤ᴵ)
\end{equation}
We only need indexed containers whose source and target indexing
types are the same, since we only need to consider endofunctors on the
category of $I$-indexed families in $𝓤$ and their algebras. Thus an
indexed signature \(Σ ≡ (I, A, B)\) is a triple that can be thought of
as a set of indices (or sorts) \(I\), an \(I\)-indexed family of operators
\(Aᵢ\), and a function mapping each operator \(a : Aᵢ\) to an
\(I\)-indexed family of arities in \(𝓤ᴵ\). For instance, the type of
vectors over some type \(D\) has an indexed signature as follows
(see~\cref{exa:length-indexed-multisets}):
\begin{itemize}
    \item natural numbers as indices;
    \item for index \(0\) a parameter-less operator for the empty vector, and
        for index \(n + 1\) an operator parametrised by \(D\) for cons; and
    \item for the operator indexed by \(0\), a family of empty types for
        arities, and for operators indexed by \(n + 1\), a family
        of arities which is the unit type just in the
        \(n\)th index, and empty otherwise.
\end{itemize}
Each indexed signature \(Σ ≡ (I, A, B)\) determines a polynomial
endofunctor \(S_Σ : 𝓤ᴵ → 𝓤ᴵ\), which is defined at a family \(X : 𝓤ᴵ\)
by
\begin{equation}%
  \label{eq:qwi-types-endo}
  (S_Σ(X))ᵢ ≝ ∑_{a : Aᵢ} (Bᵢ\,a ⇁ X) \quad\text{(for each \(i:I\))}
\end{equation}
An \(S_Σ\)-\emph{algebra} is by definition an element of the dependent product
\begin{equation}
  \Alg_Σ ≝ ∑_{X : 𝓤ᴵ} (S_Σ(X) ⇁ X)
\end{equation}
(For example, when $Σ$ is the signature for vectors over some type
$D$, then an algebra $(X,α): \Alg_Σ$ amounts to $α_0: 𝟙 → X_0$ and
$α_{n+1}: D × X_n → X_{n+1}$.)

\(S_Σ\)-algebra morphisms \((X,α)→(X',α')\)
are given by an
\(I\)-indexed family of functions \(h : X ⇁ X'\) together with a
family of elements in the
types
\begin{equation}
    (\isHom\,h)ᵢ ≝ ∏_{a : Aᵢ} ∏_{b : Bᵢ\,a ⇁ X}
    (α'ᵢ\,(a, h ∘ b) = hᵢ(αᵢ(a, b)))
\end{equation}
The WI-type \(🅆_{i:I,a : Aᵢ} Bᵢ\,a\), or \(🅆_Σ\), determined by
\(Σ ≡ (I, A, B)\) is the underlying type of an initial $S_Σ$-algebra.
The elements of \(🅆_Σ\) represent an \(I\)-indexed family of \emph{closed}
algebraic terms (i.e.~well-founded trees) over the signature \(Σ\).
Given a  family \(V :
𝓤ᴵ\), the family \(T_Σ(V):𝓤ᴵ\) of \emph{open} terms with variables from
\(V\) is the inductively defined family with constructors
\begin{equation}
  \begin{aligned}
    η &: V ⇁ T_Σ\,V\\
    σ &: S_Σ (T_Σ\,V) ⇁ T_Σ\,V
  \end{aligned}
\end{equation}
Given an \(S_Σ\)-algebra \((X, α) : \Alg_Σ\) and an \(I\)-indexed
family of functions \(f : V ⇁ X\), the unique morphism of
\(S_Σ\)-algebras from the \(S_Σ\)-algebra \((T_Σ(V),σ)\) to \((X, α)\)
that extends $f$ has an underlying family of functions \(T_Σ(V) ⇁ X\)
mapping each \(t : (T_Σ\,V)ᵢ\) to the element \(t ⋙ f\) defined
analogously to~\eqref{def_bind}.

Given an indexed signature \(Σ ≡ (I, A, B) : \Sig\), a \emph{system of
  equations} for it is an element of the dependent product
\begin{equation}%
  \label{eq:qwi-types-1}
  \Syseq_Σ ≝ ∑_{E : 𝓤ᴵ} ∑_{V : E ⥗ 𝓤ᴵ}
  \bigg( ∏_{i:I}∏_{e : Eᵢ} (T_Σ(Vᵢ\,e))ᵢ\bigg)
  × \bigg( ∏_{i:I}∏_{e : Eᵢ} (T_Σ(Vᵢ\,e))ᵢ\bigg)
\end{equation}
An \(S_Σ\)-algebra \(α : S_Σ(X) ⇁ X\) satisfies the system of equations \(ε ≡
(E,V,l,r) : \Syseq_Σ\) if for each \(i : I\) there is a proof of
\begin{equation}
  (\Sat_{α,ε}(X))ᵢ ≝ ∀(e : Eᵢ).∀(ρ : V_{i}\,e ⇁ X).\; ((l_i e) ⋙ ρ) =
  ((r_i e) ⋙ ρ)
\end{equation}

\begin{definition}[\textbf{QWI-type}]%
  \label{def:qwit}
  A \emph{QWI-type} for an indexed signature \(Σ ≡ (I, A, B) : \Sig\)
  and a system of equations \(ε ≡ (E, V, l, r) : \Syseq_Σ\) is given
  by an $I$-indexed family of types \(\QW : 𝓤ᴵ\) equipped with an
  \(S_Σ\)-algebra structure and a proof that it satisfies the
  equations
  \begin{align}
    \qwintro &: S_Σ(\QW) ⇁ \QW\label{def_qwiintro}\\
    \qwequate &: \Sat_{\qwintro,ε}(\QW)\label{def_qwiequate}
                \intertext{%
                together with functions that witness that it is an initial such
                algebra:
                }
                \qwrec &: ∏_{X : 𝓤ᴵ} ∏_{α : S_Σ(X) ⇁ X} \left(\Sat_{α,ε}(X)
                         → \left(\QW ⇁ X\right)\right)\label{def_qwirec}\\
    \qwrechom &: ∏_{X : 𝓤ᴵ}∏_{α : S_Σ(X) ⇁ X}
                ∏_{p : \Sat_{α,ε}(X)} \isHom\,(\qwrec\,X\,α\,p)\label{def_qwirechom} \\
    \qwuniq &:
              \prod_{X : 𝓤ᴵ}∏_{α : S_Σ(X) ⇁ X}∏_{p : \Sat_{α,ε}(X)}∏_{f : \QW ⇁ X}
              \isHom\,f → \qwrec\,X\,α\,p = f\label{def_qwiuniq}
  \end{align}
\end{definition}


\FloatBarrier
\section{Weakly initial sets of covers}%
\label{sec:wisc}

In \Cref{sec:qw-types} we defined a QWI-type in a topos to be an
initial algebra for a given (possibly infinitary) indexed signature
and system of equations (\Cref{def:qwit}). If one interprets these
notions in the topos of sets in classical Zermelo-Fraenkel set theory
with the Axiom of Choice (ZFC), one regains the usual notion from
universal algebra of initial algebras for equational theories that are
multi-sorted~\parencite{BirkhoffG:heta} and
infinitary~\parencite{LintonFEJ:aspeq}. Thus in ZFC, the QWI-type for
a signature $Σ=(I,A,B)$ and system of equations $ε=(E,V,l,r)$ can be
constructed by first forming the $I$-indexed family $🅆 _Σ$ of sets of
well-founded trees over $Σ$ and then quotienting by the congruence
relation $∼_ε$ on $🅆 _Σ$ generated by $ε$. The $I$-indexed family of
quotient sets $(🅆 _Σ)/{∼_ε}$ yields the desired initial algebra for
$(Σ,ε)$ provided the $S_Σ$-algebra structure on $🅆 _Σ$ induces one on
the quotient sets. It does so, because for each operator in the
signature, using the Axiom of Choice (AC) one can pick representatives
of the (possibly infinitely many) equivalence classes that are the
arguments of the operator, apply the interpretation of the operator in
$🅆 _Σ$ and then take the equivalence class of that. So the topos of
sets in ZFC has QWI-types.

Is this use of AC really necessary?
\textcite[Section~9]{BlassA:worfac} shows that if one drops AC and
just works in ZF, then provided a certain large cardinal axiom is
consistent with ZFC, it is consistent with ZF that there is an
infinitary equational theory with no initial algebra. He shows this by
first exhibiting a countably presented equational theory whose initial
algebra has to be an uncountable regular cardinal; and secondly
appealing to the construction of \textcite{GitikM:unccs} of a model of
ZF with no uncountable regular cardinals (assuming a certain large
cardinal axiom). \textcite{lumsdaine2019semantics} turn the infinitary
equational theory of Blass into a higher-inductive type that cannot be
proved to exist in ZF (and hence cannot be constructed in type theory
just using pushouts and the natural numbers). We show in
\Cref{blass_lumsdaine_shulman_type} that this higher inductive type
can be presented as a QWI-type.

So one cannot hope to construct QWI-types using a type theory which is
interpretable in just ZF\@. However, the type theory in which we work
(\Cref{sec:type-theory}) already requires going beyond ZF to be able
to give a classical set-theoretic interpretation of its universes (by
assuming the existence of enough strongly inaccessible cardinals, for
example). So the above considerations about non-existence of initial
algebras for infinitary equational theories in ZF do not necessarily
rule out the construction of QWI-types in other toposes with natural
numbers object and universes. In \Cref{sec:construction}, we show that
QWI types exist in a topos if it satisfies a weak form of choice
principle that we call \IWISC{} and which is known to hold for a wide
range of toposes.  \IWISC{} is a version for universes in the type
theory in which we are working of a property introduced by
\textcite{vandenberg2014axiom}, who work in constructive set theory
CZF and call it the \emph{Axiom of Multiple Choice} (building upon
previous work by \textcite{MoerdijkI:typttc}); it is related to the
\emph{Type Theoretic Collection Axiom} \textsc{TTCA}$_f$ of
\textcite{StreicherT:reamc} (see \Cref{rem:global-wisc}).
\begin{definition}[\textbf{Indexed WISC}]%
  \label{def:wisc}
  A \emph{family} in a universe $𝓤$ is just a pair $(A,B)$ where $A:𝓤$
  and $B:A→𝓤$. A family $(C,E)$ in $𝓤$ is a \emph{wisc} for $Y:𝓤$ if
  for all $X:𝓤$ and surjective functions $q:X → Y$ in $𝓤$, there
  exists $c:C$ and $f:E\,c → X$ for which $q∘f: E\,c→ Y$ is
  surjective. A family $(C,E)$ is a wisc for a family $(A,B)$ if it is
  a wisc for each type $B\,a$ in the family.  We say that \emph{the
  universe $𝓤$ satisfies \IWISC{}} if for all families $(A,B)$ in $𝓤$,
  there exists a family in $𝓤$ which is a wisc for $(A,B)$.
\end{definition}
Following [\url{https://ncatlab.org/nlab/show/WISC}], ``wisc'' stands
for ``weakly initial set of covers'', with the terminology being
justified as follows.  A \emph{cover} of $Y:𝓤$ is just a surjective
function $q:X → Y$ for some $X:𝓤$. If $(C,E)$ is a wisc for $Y$ as in
the above definition, then the family of all covers of the form
$E\,c → Y$ with $c:C$ is weakly initial, in the sense that for
every cover $q:X → Y$ in $𝓤$, there is a cover in the family
that factors through $q$.

Note that in classical set theory AC implies \IWISC;\@ this is because
the Axiom of Choice implies that covers split and hence any family
$(A,B)$ is a wisc for itself. \IWISC{} is independent of ZF in
classical logic~\parencite{KaragilaA:embocw,RobertsDM:weacpw}. The
results of \textcite{vandenberg2014axiom} imply that in constructive
logic it is preserved by many ways of making new toposes from existing
ones, in particular by the formation of sheaf toposes and
realizability toposes over a given topos. Thus it holds in all the
toposes commonly used for the semantics of Type Theory. It is in this
sense that \IWISC{} is a constructively acceptable form of the axiom
of choice. (\textcite{RobertsDM:weacpw} gives examples of toposes
which fail to satisfy it.)

\begin{remark}%
  \label{rem:global-wisc}
  \IWISC{} is a property of a single universe, in contrast to
  Streicher's type theoretic formulation of a WISC axiom,
  \textsc{TTCA}$_f$~\parencite{StreicherT:reamc}, which uses two
  universes (and which can be shown to imply IWISC). With a single
  universe it seems necessary to quantify over indexed families in the
  universe (hence our choice of terminology), rather than just over
  types in the universe as in \emph{loc.~cit.}; see
  \textcite[Section~5.1]{LevyP:broigp}.  Note that \IWISC{} asks that
  a wisc merely exists for each family in $𝓤$. A stronger requirement
  on $𝓤$ would be to ask for a function assigning a wisc to each
  family; and this is equivalent to asking for a function mapping each
  $B:𝓤$ to a wisc for $B$. (For note that, given a family $(A,B)$, if
  we have a function mapping each $a:A$ to a wisc $(C_a,E_a)$ for
  $B\,a$, then the single family $(C,E)$ with $C=∑_{a:A}C_a$ and
  $E(a,c) = E_a\,c$ is a wisc for all the $B\,a$ simultaneously.)
  \textcite{LevyP:broigp} calls this property \GWISC{}. Although the
  topos of sets in ZFC with Grothendieck universes satisfies \GWISC{}, it
  is not known which other toposes do.
\end{remark}
The following lemma gives a convenient reformulation of the wisc
property from \Cref{def:wisc} that we use in the next section.
\begin{lemma}%
  \label{lem:wisc}
  Suppose that $A:𝓤$ has a wisc $(C:𝓤,E:C→𝓤)$.  If $B:A→𝓤$ and
  $φ:∏_{x:A}(B\,x→\Prop)$ satisfy $∀(x:A).∃(y:B\,x).\; φ\,x\,y$, then
  there exist $c:C$, a surjection $p:E\,c→ A$ and a function
  $q:\prod_{z:E\,c}B(p\,z)$ satisfying
  $∀(z:E\,c).\;φ\,(p\,z)\,(q\,z)$.
\end{lemma}
\begin{proof}
  Consider the type $Φ\defeq∑_{x:A}∑_{y:B\,x}\,φ\,x\,y$ in $𝓤$. By assumption
  $π_1:Φ→A$ is a surjection, so since $(C,E)$ is a wisc for $A$, there exist
  $c:C$ and $e:E\,c → Φ$ with $π_1∘e$ a surjection. So we can take $p ≝ π_1∘e$
  and $q ≝ π₁∘π₂∘e$; and then $π₂(π₂(e\,z))$ is a proof that $φ\,(p\,z)\,(q\,z)$
  holds.
\end{proof}


\FloatBarrier
\section{Size}%
\label{sec:size}

\textcite{swan2018wtypes} uses \IWISC{} to construct \emph{W-types
  with reductions}, which are a simple special case of QWI-types (see
\Cref{exa:wsusp}). We will see that in fact \IWISC{} implies the
existence of all QWI-types, but using a different approach from that
of \textcite{swan2018wtypes}. We show in this section that, starting
with a given signature $Σ$ and system of equations $ε$
(\cref{def:syse}), using \IWISC{} one can construct a well-founded
type of ``sizes'' with the property that colimits of diagrams indexed
by such sizes are preserved when taking exponents by the arity types of $Σ$
and $ε$. This will enable us to construct the QW-type for $(Σ,ε)$ as a
colimit in \Cref{sec:construction}. We consider the general case of
indexed signatures in \cref{rem:indexed-cocontinuity}; and the
construction of QWI-types as a sized-indexed colimit is detailed in
the accompanying Agda
formalization~\parencite{fiore2021agdacode}. Size is here playing a
role in the constructive type theory of toposes that in classical set
theory would be played by ordinal numbers. The next definition gives
the properties of size that we need.

\begin{definition}[\textbf{Size}]%
  \label{def:size}
  A type $\Size$ in a universe $𝓤$ will be called a \emph{type of
    sizes} if it comes equipped with
  \begin{itemize}
  \item a binary relation $\_<\_:\Size → \Size → \Prop$ which
    is transitive
    \begin{equation}%
      \label{eq:<trans}
      ∀ i,j,k.\; i < j → j < k → i <k
    \end{equation}
    and well-founded
    \begin{equation}%
      \label{eq:<iswf}
      ∀(φ : \Size→\Prop).(∀ i. (∀ j.\;{j<i} → φ\,j) → φ\,i) → ∀ i.\;φ\,i
    \end{equation}

  \item a distinguished element $0ˢ$

  \item a binary operation $\_⊔ˢ\_:\Size→\Size→\Size$ which gives
    upper bounds with respect to $<$ for pairs of elements
    \begin{equation}%
    \label{eq:<ub}
    ∀ i,j.\; i < i ⊔ˢ j  \;∧\; j < i ⊔ˢ j
  \end{equation}
  \end{itemize}
  Note that defining $\Suc{i} \defeq i ⊔ˢ i$ we get a form of successor
  operation satisfying
  \begin{equation}
    ∀ i.\; i < \Suc{i}\label{eq:<suc}
  \end{equation}
  For the examples of types of sizes given in \Cref{exa:plump} this
  successor operation preserves $<$, but we do not need that property
  for the results in this paper.\footnote{In fact we do not need the
    distinguished element $0ˢ$ either, but it does no harm to require
    it and then sizes are directed with respect to $<$ in the usual
    sense of having upper bounds for all finite subsets, including
    the empty one.}
\end{definition}
Such a type of sizes has many properties of the classic notion of
limit ordinal, except that we do not require the order to be total
($∀i, j.\; i<j ∨ i=j ∨ j <i$); that would be too strong in a
constructive setting and indeed does not hold in the examples
below. Nor do we have any need for an extensionality property
($∀ i,j.\;(∀ k.\;k < i ↔ k< j) → i = j$). In the precursor to this
paper~\parencite{fiore2020constructing} we made use of Agda's built in
type $\Size$ and its version of sized
types~\parencite{AbelA:typbti}. This shares only some of the
properties of the above definition. In particular, it has a
``stationary'' size $∞$ satisfying $∞ < ∞$ and hence the relation $<$
for Agda's notion of size is not well-founded.\footnote{Unfortunately
  in the current version of Agda (2.6.2) it is possible to prove
  well-foundedness of $<$ for its built in type of sizes $\Size$, and
  this leads immediately to logical unsoundness; see
  \url{github.com/agda/agda/issues/3026}. For this reason we avoid the
  use of Agda's sized types in the current Agda development
  accompanying this paper.}

\begin{notation}[\textbf{Bounded quantification over sizes}]
  We write $∏_{j<i}\,\_$ for $∏_{j:(↓ i)}\,\_$, where
  \begin{equation}%
    \label{eq:down-seg}
    ↓ i \defeq \{ j : \Size ∣ j<i \}
  \end{equation}
  is the type of sizes below $i:\Size$. Similarly for other binders
  involving sizes, such as $∀{j<i}.\;\_$ and $ƛ{j<i}·\_$.
\end{notation}
In the next section we need the fact that well-founded recursion~\eqref{eq:wfrec} can be reduced to well-founded induction~\eqref{eq:<iswf} by defining the graph of the recursive function and
then appealing to unique choice~\eqref{eq:uniquechoice} and function
extensionality~\eqref{eq:funext}:
\begin{proposition}[\textbf{Well-founded recursion}]%
  \label{prop:wfrec}
  For any family of types $A:\Size→ 𝓤$ and function
  $α:∏_i\left(\left(∏_{j<i}\,A\,j\right) → A\,i\right)$,
  there is a unique function $\Rec A\,α : ∏_i\,A\,i$
  satisfying
  \begin{equation}%
    \label{eq:wfrec}
    ∀i.\; \Rec A\,α\,i = α\,i\,(ƛ{j< i}· \Rec A\,α\,j)
  \end{equation}
\end{proposition}
\begin{proof}
  Let $R : ∏_i (A\,i → \Prop)$ be the least\footnote{Instead of the
    impredicative definition of $R$ (and $<$ and $≤$ in
    \Cref{exa:plump}) as the least relation closed under some rules,
    in our Agda development such relations are constructed as
    inductively defined datatypes in the universe $\Prop$, allowing
    one to use dependent pattern-matching to simplify proofs about
    them.} family of relations satisfying
  \begin{equation}%
    \label{eq:wfrec-prf}\textstyle
    ∀i.∀(f:∏_{j<i}\,A\,j).\;(∀{j<i}.\;R_j(f\,j)) → R_i(α\,i\,f)
  \end{equation}
  The fact that $R$ is least with this property allows one to prove
  $∀ i.∃! (x:Aᵢ).\;R_i\,x$ by applying~\eqref{eq:<iswf} with
  $φ\,i \defeq ∃!\,(x:Aᵢ).\;R_i\,x$. So by unique
  choice~\eqref{eq:uniquechoice} there is $r: ∏_i\,A\,i$ satisfying
  $∀i.\; R_i(r\,i)$ and hence by~\eqref{eq:wfrec-prf} also satisfying
  $∀ i.\;R_i(α\,i\,(ƛ{j<i}· r\,j))$. Since for each $i$ there is a
  unique $x$ with $R_i\,x$, it follows that
  $r\,i = α\,i\,(ƛ{j<i}· r\,j)$.  So we can take $\Rec A\,α$ to be
  $r$.  If $r':∏_i\,A\,i$ also satisfies
  $∀i.\,r'i =α\,i\,(ƛ{j<i}·r'j)$, then $∀i.\,r'i= r\,i$ follows from
  another application of~\eqref{eq:<iswf}, so that $r'=r$ by function
  extensionality~\eqref{eq:funext}.
\end{proof}

\begin{example}[\textbf{Plump order}]%
  \label{exa:plump}
  Let $\Size:𝓤$ be the W-type determined by some type $\Op:𝓤$ and
  family $\Ar:\Op→𝓤$ (see \cref{sec:w-types}). Thus
  $\Size$ is inductively defined with constructor
  $σ:(∑_{a:\Op}(\Ar\,a → \Size)) → \Size$. The \emph{plump} ordering~\parencite{TaylorP:intso} on $\Size$ is given by the
  least relations ${\_<\_}, {\_≤\_}: \Size → \Size → \Prop$ satisfying for
  all $a:\Op$, $b:\Ar\,a → \Size$ and $i: \Size$
  \begin{align}
    &(∀(x:\Ar\,a).\;b\,x < i) → σ\,a\,b ≤ i\label{eq:plump1}\\
    &(∃(x:\Ar\,a).\;i≤ b\,x) → i < σ\,a\,b\label{eq:plump2}
  \end{align}
  It is not hard to see that the relation $<$ is
  transitive~\eqref{eq:<trans} and
  well-founded~\eqref{eq:<iswf}. Furthermore one can prove that $≤$ is
  reflexive and hence from~\eqref{eq:plump2} we get
  \begin{equation}%
    \label{eq:plump3}
    ∀(x:\Ar\,a).\; b\,x < σ\,a\,b
  \end{equation}
  So for each operation $a:\Op$, every family $b:\Ar\,a → \Size$ of
  elements of $\Size$ indexed by its arity $\Ar\,a$ has an upper bound
  with respect to $<$, namely $σ\,a\,b$. In particular, if the
  signature $(\Op,\Ar)$ contains an operation $a₂:\Op$ whose arity is
  $\Ar\,a₂ = 𝟚$, then an upper bound for any $i,j: \Size$ with respect
  to $<$ is given by
  \begin{equation}%
    \label{eq:ub}
    i ⊔ˢ j \defeq σ\,a₂\,b_{i,j}
  \end{equation}
  with $b_{i,j}:𝟚 → \Size$ defined by $b_{i,j}\,0 = i$ and
  $b_{i,j}\,1 = j$. So~\eqref{eq:<ub} is satisfied in
  this case. Similarly, if the signature $(\Op,\Ar)$ contains an
  operation $a₀:\Op$ whose arity is $\Ar\,a₀ = 𝟘$, then $\Size$
  contains a distinguished element
  \begin{equation}%
    \label{eq:sizezero}
    0ˢ \defeq σ\,a₀\,b
  \end{equation}
 with $b:𝟘→\Size$ uniquely determined. Therefore, we have:
\end{example}

\begin{lemma}%
  \label{lem:plump}
  If the signature $\Op:𝓤,\Ar:\Op→𝓤$ contains a nullary operation
  ($a₀:\Op$ with $\Ar\,a₀ = 𝟘$) and a binary operation ($a_2:\Op$ with
  $\Ar\,a₂=𝟚$), then the corresponding W-type $\Size:𝓤$, equipped with
  the plump order $<$, is a type of sizes in the sense of
  \Cref{def:size}. Furthermore for any $a:\Op$, every
  $(\Ar\,a)$-indexed family of sizes $b:\Ar\,a → \Size$ has an upper
  bound with respect to $<$ (namely $σ\,a\,b$). \qed
\end{lemma}
Given a signature $Σ$ in a universe $𝓤$ satisfying \IWISC{}, the lemma
allows us to prove the existence of a type of sizes $\Size$ with
enough upper bounds to be able to prove a cocontinuity property of the
polynomial endofunctor associated with $Σ$
(\Cref{cor:poly-cocontinuity} below); we apply this in
\Cref{sec:construction} to construct QW-types for a given signature
and system of equations (and QWI-types for indexed signatures and
systems of equations as defined in \cref{sec:qwi-types}). The
cocontinuity property has to do with colimits of $\Size$-indexed
diagrams in $𝓤$. To state it, we first recall some
semi-standard\footnote{They are only \emph{semi}-standard, because $<$
is not reflexive (indeed is irreflexive, because of well-foundedness),
so that $(\Size,<)$ is only a (thin) \emph{semi}-category.}
category-theoretic notions to do with diagrams and colimits.

\subsection{Colimits of size-indexed diagrams}%
\label{sec:colimit}

By definition, given a type of sizes $\Size$ in a universe $𝓤$
(\Cref{def:size}), a \emph{$\Size$-indexed diagram} is given by a
family of types $D: \Size→ 𝓤$ equipped with functions
$δ_{i,j}:D_i → D_j$ for all $i,j:\Size$ with $i < j$, satisfying
\begin{equation}%
  \label{eq:diag}
  ∀(i,j,k : \Size).\; i<j<k → δ_{i,k} = δ_{j,k}∘ δ_{i,j}
\end{equation}
As usual, the \emph{colimit} of such a diagram is a cocone of functions in
$𝓤$
\begin{equation}%
  \label{eq:colim-cocone}
  (ν_D)_i : D_i → \colim D \qquad
  ∀(i, j : \Size).\;i < j → (ν_D)_i = (ν_D)_j∘δ_{i,j}
\end{equation}
with the universal property that for any other cocone
\[
  f_i : D_i → C \qquad
  ∀(i, j : \Size).\;i < j → f_i = f_j∘δ_{i,j}
\]
there is a unique function $f:\colim D → C$ satisfying $∀ i.\; f_i = f
∘ (ν_D)_i$.

Colimits can be constructed using quotient types: we define a binary
relation $\_{∼}\_$ on $∑_i\, D_i$ by:
\begin{equation}%
  \label{eq:colim-construction}
  (i,x) ∼ (j,y) \;\defeq\;
  ∃(k : \Size).\; {i < k} \;∧\; {j < k} \;∧\;
  {δ_{i,k}\,x = δ_{j,k}\,y}
\end{equation}
This is an equivalence relation, because $(\Size,{<})$ has
property~\eqref{eq:<ub}.  Quotienting by it
yields
\begin{equation}%
  \label{eq:colim}
  \textstyle
  \colim D \;\defeq\; \left(∑_i D_i\right)/{∼}
\end{equation}
with the universal cocone functions $(ν_D)_i$ given by mapping each
$x:D_i$ to the equivalence class $[i,x]_{∼}$.
\begin{definition}[\textbf{Preservation of colimits by taking a
    power}]%
  \label{def:pres-colim-powers}
  Given a $\Size$-indexed diagram $(D,δ)$ in $𝓤$ and a type $X:𝓤$, we
  get a \emph{power} diagram $(D^X,δ^X)$ in $𝓤$ with
  \begin{equation}%
    \label{eq:pres-colim-powers}
    \begin{aligned}
      (D^X)_i &\defeq X → D_i &&(i: \Size)\\
      (δ^X)_{i,j} &\defeq ƛ(f:X → D_i)· δ_{i,j}∘ f &&(i,j: \Size)
    \end{aligned}
  \end{equation}
  Post-composition with $(ν_D)_i$ gives a cocone under the diagram $D^X$
  with vertex $X → \colim D$ and by universality this induces a function
  \begin{equation}%
    \label{eq:comparison}
    \begin{aligned}
      &κ_{D,X}: \colim(D^X) → (X → \colim D)\\
      &κ_{D,X}\,[i,f]_{∼} = (ν_D)_i∘ f
    \end{aligned}
  \end{equation}
  One says that \emph{taking a power by $X:𝓤$ preserves
    $\Size$-indexed colimits} if for all $\Size$-indexed
  diagrams $(D,δ)$ in $𝓤$ the function $κ_{D,X}$ is an isomorphism.
\end{definition}

\begin{theorem}[\textbf{Cocontinuity}]%
  \label{thm:cocontinuity}
  Suppose $𝓤$ is a universe satisfying \IWISC{} in a topos with
  natural numbers object. Given $A:𝓤$ and $B: A → 𝓤$, there exists a
  type of sizes $\Size : 𝓤$ with the property that for all $a:A$,
  taking a power by the type $B\,a : 𝓤$ preserves $\Size$-indexed
  colimits.
\end{theorem}
\begin{proof}
  The proof is in three parts,~\ref{cocont1}--\ref{cocont3}. In
  part~\ref{cocont1} we show how to find a suitable type of sizes
  $\Size$ using \Cref{lem:plump}; so $\Size$ is a W-type for a
  suitable signature derived from $A:𝓤$ and $B: A → 𝓤$. The signature
  contains arities arising from a wisc whose existence is guaranteed
  by \IWISC. In fact we need to consider not only the covers in such a 
  wisc, but also a wisc for the family of (subtypes of) covers in this
  wisc, for the same reason that Swan uses ``2-cover
  bases''~\parencite[Section~3.2.2]{swan2018wtypes}.  Next, we have to
  prove that~\eqref{eq:comparison} is an isomorphism when $X:𝓤$ is
  $B\,a$, for any $a:A$; and for this it suffices to prove that it is~\ref{cocont2} injective and~\ref{cocont3} surjective. For then
  we can apply unique choice~\eqref{eq:uniquechoice} to construct a
  two-sided inverse for $κ_{D,X}$. By definition of
  $∼$~\eqref{eq:colim-construction}, $κ_{D,X}$ is injective iff
  \begin{multline}%
    \label{eq:injective}
    ∀(i, j : \Size). ∀(f:X →  D_i). ∀(f':X → D_j).\; (ν_D)_i∘ f
    = (ν_D)_j∘ f' → {}\\
    ∃(k : \Size).\; {i<k} \;∧\; {j<k} \;∧\; {δ_{i,k}∘ f =
      δ_{j,k}∘ f'}
  \end{multline}
  and surjective iff
  \begin{equation}%
    \label{eq:surjective}
    ∀(f:X → \colim D).∃(i : \Size).∃(f':X → D_i).\; (ν_D)_i∘ f'
    = f
  \end{equation}
  \begin{enumerate}[label={(\arabic*)}]
  \item\label{cocont1} \emph{Construction of $\Size$}: Given $A:𝓤$
    and $B: A → 𝓤$,  using \IWISC{} let $(C,F)$ be a
    wisc for the family $(A,B)$. For each $c:C$, $a:A$ and
    function $f:F\,c → B\,a$ we can form the kernel of $f$
    \begin{equation}%
      \label{eq:ker}
      \textstyle
      \Ker f \defeq ∑_{x,x': F\,c}(f\,x = f\,x')
    \end{equation}
    Applying \IWISC{} again, let $(C',F')$ be a wisc for this family
    of kernels indexed by $(c, a, f) : ∑_{(c,a):C × A}(F\,c →
    B\,a)$. Finally, consider the signature with
    $\Op = 𝟙 + 𝟙 + A + C + C'$ and $\Ar:\Op→𝓤$ the function mapping
    $0:𝟙$ in the first summand of $\Op$ to $𝟘$, $0:𝟙$ in the second
    summand to $𝟚$, $a:A$ in the third summand to $B\,a$, $c:C$ in the
    fourth summand to $F\,c$ and $c':C'$ in the fifth summand to
    $F'c'$. Then as in \Cref{lem:plump}, the W-type for the signature
    $(\Op,\Ar)$ is a type of sizes. Call it $\Size$.

  \item\label{cocont2} \emph{Proof of~\eqref{eq:injective}}: Recall
    that $X:𝓤$ is of the form $B\,a$ for some
    $a:A$. Suppose we have $f:X → D_i$ and $f':X → D_j$ satisfying
    $(ν_D)_i∘ f = (ν_D)_j∘ f'$. So by definition of $\colim D$~\eqref{eq:colim} we have
    \begin{equation}
      ∀(x:X).∃(k : \Size).\; i<k \;∧\; j<k\;∧\;
      δ_{i,k}(f\,x) = δ_{j,k}(f'\, x)
    \end{equation}
    Using the version of the wisc property of $(C,F)$
    from \Cref{lem:wisc}, there is $c:C$, $p:F\,c → X$ and
    $s:F\,c→ \Size$ with $p$ surjective and satisfying
    \begin{equation}%
      \label{eq:injective-proof}
      ∀(z:F\,c).\; i<s\,z \;∧\; j<s\,z \;∧\; δ_{i,s\,z}(f(p\,z)) =
      δ_{j,s\,z}(f'(p\,z))
    \end{equation}
    Since $F\,c$ is one of the arities in the signature of the W-type
    $\Size$, by \Cref{lem:plump} we have that
    $s:F\,c → \Size$ has an upper bound, i.e.~there is
    $k: \Size$ with $∀(z: F\,c).\;s\,z < k$; and by~\eqref{eq:<ub} we can assume further that $i< k$ and $j <
    k$. Furthermore, from~\eqref{eq:injective-proof} and~\eqref{eq:diag}
    we get
    $∀(z:F\,c).\; δ_{i,k}(f(p\,z)) =
    δ_{j,k}(f'(p\,z))$; but $p$ is surjective, so
    $δ_{i,k}∘f = δ_{j,k}∘f'$, as required for~\eqref{eq:injective}.

  \item\label{cocont3} \emph{Proof of~\eqref{eq:surjective}}: Suppose
    we have $f:X → \colim D$. Since the quotient function
    $[\_]_{∼}: Dᵢ → \colim D$ is surjective, using \Cref{lem:wisc}
    again, there is $c:C$, $p:F\,c → X$, $s:F\,c→ \Size$ and
    $g:∏_{z:F\,c}D_{s\,z}$ with $p$ surjective and satisfying
    $∀(z:F\,c).\; f(p\,z) = [ s\,z , g\,z]_{∼}$. As before, we have
    that $s:F\,c→ \Size$ has an upper bound, i.e.~there is $j: \Size$
    with $∀(z: F\,c).\;s\,z < j$. So we get a function $g': F\,c → Dⱼ$
    by defining $g'\,z \defeq δ_{s\,z,j}(g\,z)$ and hence
    \begin{equation}%
      \label{eq:surj-prf-1}
      f(p\,z) = [ s\,z , g z]_{∼} = [ j, δ_{s\,z,j}(g
      z)]_{∼} = (ν_D)ⱼ(g'z)
    \end{equation}
    Let $Y\defeq \Ker\,p$ be the kernel of $p$ as
    in~\eqref{eq:ker}. Then for any $(z,z',\_):Y$, since $p\,z =
    p\,z'$,  we have
    $(ν_D)ⱼ(g'z) = f(p\,z) = f(p\,z') =
    (ν_D)ⱼ(g'z')$ and hence
    $∃ k.\; j<k ∧ {δ_{j,k}(g'z) =
      δ_{j,k}(g'z')}$. So we can apply \Cref{lem:wisc} again
    to deduce the existence of $c':C'$, $⟨p₁',p₂',\_⟩:F'c'→ Y$
    and $s':F'c'→ \Size$ with $⟨p₁',p₂'⟩$ surjective and
    satisfying
    \begin{equation}%
      \label{eq:surj-proof1}
      ∀(z':F'c').\; {j<s'z'} \;∧\; δ_{j,s'z'}(g'(p₁'\,z')) =
      δ_{j,s'z'}(g'(p₂'\,z'))
    \end{equation}
    Since $F'c'$ is one of the arities in the signature of the W-type
    $\Size$, by \Cref{lem:plump} we have that
    the family of sizes $s':F'c'→ \Size$ has
    an upper bound, $i$ say; and by~\eqref{eq:<ub} we can assume
    $j < i$. Let $g'': F\,c → Dᵢ$ be
    $δ_{j,i} ∘ g'$. Thus from~\eqref{eq:surj-prf-1} we~have
    \begin{equation}\label{eq:surj-prf-2}
      f(p\,z) = (ν_D)ᵢ(g''z)
    \end{equation}
    and from~\eqref{eq:surj-proof1} also  $g''∘ p_1' = g''∘ p₂'$.
    So altogether we have functions
    \begin{equation}
      \xymatrix{{F\,c'}\ar[r]^{⟨p₁',p₂'⟩} & Y \ar@<1ex>[r]^{π_1}
        \ar@<-1ex>[r]_{π₂} & {F\,c} \ar[r]^{g''} & {D_i}}
    \end{equation}
    with
    $g''∘ π₁∘ ⟨p₁',p₂'⟩ = g'' ∘ p_1' = g''∘p_2'=g'' ∘ π_2 ∘
    ⟨p₁',p₂'⟩$. Since $⟨p₁',p₂'⟩$ is surjective, this implies
    $g''∘ π₁ = g'' ∘ π_2$. But since $Y=\Ker\, p$, $p$ is the
    coequalizer of $π_1$ and $π_2$ (since as we noted in
    \cref{sec:type-theory}, toposes have effective epimorphisms);
    therefore there is a unique function $f':X →D_i$ satisfying
    $f'∘p = g''$. Since $(ν_D)_i ∘ f' ∘ p = (ν_D)_i ∘ g'' = f ∘ p$
    by~\eqref{eq:surj-prf-2} and $p$ is surjective, we have
    $(ν_D)_i ∘ f' = f$, completing the proof of~\eqref{eq:surjective}
    and hence of \Cref{thm:cocontinuity}.
    \qedhere
  \end{enumerate}
\end{proof}

\begin{remark}%
  \label{rem:cocontinuity}
  Note that the type of sizes constructed in the proof of the theorem
  has the property that for any $a:A$, $\Size$ has upper bounds (with
  respect to $<$) for families of sizes indexed by $B\,a$ (because by
  construction $B\,a$ is one of the arities in the signature of the
  W-type $\Size$ and so \cref{lem:plump} applies). Such upper bounds
  are not needed in the above proof. We included them because they are
  needed below in the proof of \cref{thm:main} when proving property~\eqref{eq:qwintro2}.
\end{remark}
\begin{definition}[\textbf{Preservation of colimits by polynomial
    endofunctors}]%
  \label{def:pres-colim-polyendofunc}
  If $(D,δ)$ is a $\Size$-indexed diagram in $𝓤$, then we get
  another such, $(S_Σ∘D, S_Σ∘δ)$, by composing with the polynominal
  endofunctor $S_Σ:𝓤 → 𝓤$ associated with the signature $Σ$ as
  in~\eqref{def_S_endofunctor}:
  \begin{equation}
    \begin{aligned}
      (S_Σ∘D)_i    &\defeq S_Σ(D_i) &&(i: \Size)\\
      (S_Σ∘δ)_{i,j} &\defeq S_Σ(δ_{i,j}) &&(i,j: \Size)
    \end{aligned}
  \end{equation}
  Applying $S_Σ$ to~\eqref{eq:colim-cocone}
  gives a cocone under $(S_Σ∘D, S_Σ∘δ)$ with vertex $S_Σ(\colim D)$ and
  this induces a function in $𝓤$,
  \begin{equation}%
    \label{eq:comparison-poly}
    κ_{D,Σ} : \colim(S_Σ∘D) → S_Σ(\colim D)
  \end{equation}
  One says that \emph{the polynomial endofunctor $S_Σ$ preserves
    $\Size$-indexed colimits} if $κ_{D,Σ}$ is an isomorphism for all
  diagrams $(D,δ)$.
\end{definition}

\begin{corollary}[\textbf{Cocontinuity of $S_Σ$}]%
  \label{cor:poly-cocontinuity}
  In a topos with natural numbers object, given a signature
  $Σ=(A:𝓤, B:A→𝓤)$ in a universe $𝓤$ satisfying \IWISC{}, there exists a
  type of sizes $\Size$ with the property that the associated
  polynominal endofunctor $S_Σ:𝓤 → 𝓤$ preserves $\Size$-indexed
  colimits.
\end{corollary}
\begin{proof}
  We apply \Cref{thm:cocontinuity} to find a type of sizes, $\Size$
  for the given $A:𝓤$ and $B:A → 𝓤$. So for each diagram $(D,δ)$ in
  $𝓤$, we have properties~\eqref{eq:injective} and~\eqref{eq:surjective} when $X$ is $B\,a$, for any $a:A$.  The
  function~\eqref{eq:comparison-poly} satisfies for all $i:\Size$,
  $a:A$ and $b:B\,a → D_i$
  \begin{equation}
    κ_{D,Σ} [i,(a,b)]_{∼} = (a, (ν_D)_i∘ b)
  \end{equation}
  We have that $κ_{D,Σ}$ is injective, because if
  $(a, (ν_D)_i∘b) = (a', (ν_D)_j∘b')$, then $a=a'$ and
  $(ν_D)_i∘b=(ν_D)_j∘b'$; but then from~\eqref{eq:injective} with
  $X=B\,a=B\,a'$, there is some $k$ with $i,j<k$ and
  $δ_{i,k}∘b = δ_{j,k}∘b'$, so that $[i,(a,b)]_{∼}$ and
  $[ j,(a',b')]_{∼}$ are equal terms of $\colim(S_Σ∘D)$.

  Furthermore, $κ_{D,Σ}$ is surjective because given
  $(a,b):S_Σ(\colim D)$, from~\eqref{eq:surjective} with $X=B\,a$,
  there exist $i:\Size$ and $b':B\,a → D_i$ with $(ν_D)_i∘b' = b$; so
  $[i,(a,b')]_{∼}$ in $\colim(S_Σ∘D)$ is mapped by $κ_{D,Σ}$ to
  $(a,b)$.

  Since $κ_{D,Σ}$ is both injective and surjective, we can apply
  unique choice~\eqref{eq:uniquechoice} to conclude that it is an
  isomorphism.
\end{proof}

\begin{remark}[\textbf{Cocontinuity for polynomial endofunctors on
    $𝓤^I$}]%
  \label{rem:indexed-cocontinuity}
  There are indexed versions of \cref{thm:cocontinuity} and
  \Cref{cor:poly-cocontinuity}. To state them for an indexing type
  $I:𝓤$, we need to consider $\Size$-indexed diagrams in $𝓤^I$ and
  their colimits. Since the latter are given pointwise by colimits in
  $𝓤$, this makes what follows a simple extension of the previous
  development.

  Given a $\Size$-indexed diagram $(D,δ)$ in $𝓤ᴵ$
  and a family $X:𝓤ᴵ$, the indexed version of~\eqref{eq:pres-colim-powers} is
  a \emph{power} diagram $(D^X,δ^X)$ in $𝓤$ with
  \begin{equation}
    \begin{aligned}
      (D^X)_i &\defeq X ⇁ D_i &&(i: \Size)\\
      (δ^X)_{i,j} &\defeq ƛ(f:X ⇁ D_i)· δ_{i,j}∘ f &&(i,j: \Size)
    \end{aligned}
  \end{equation}
  Post-composition with $(ν_D)_i$ gives a cocone under the diagram
  $D^X$ with vertex $X ⇁ \colim D$ and this induces a
  function which is the indexed version of~\eqref{eq:comparison}:
  \begin{equation}
    \begin{aligned}
      &κ_{D,X}: \colim(D^X) ⇁ (X ⇁ \colim D)\\
      &κ_{D,X}\,[i,f]_{∼} = (ν_D)_i∘ f
    \end{aligned}
  \end{equation}
  One says that \emph{taking a power by $X:𝓤ᴵ$ preserves $\Size$-indexed
    colimits} if for all $\Size$-indexed diagrams $(D,δ)$ in $𝓤ᴵ$ the
  function $κ_{D,X}$ is an isomorphism. Then we have:
  \begin{quote}
    \emph{If $𝓤$ is a universe satisfying \IWISC{} in a topos with natural
    numbers object, then given $A,I:𝓤$ and $B: A → 𝓤ᴵ$, there exists
    a type of sizes $\Size : 𝓤$ with the property that for all $a:A$,
    taking a power by a family $B\,a : 𝓤ᴵ$ preserves $\Size$-indexed
    colimits.}
  \end{quote}
  The proof of this is similar to the proof of \cref{thm:cocontinuity}
  and can be found in the Agda development accompanying this
  paper~\parencite{fiore2021agdacode}.

  Given a $\Size$-indexed diagram $(D,δ)$ in $𝓤ᴵ$, and an $I$-indexed
  signature $Σ$ as in~\eqref{eq:qwi-types-endo}, then we get another
  \(\Size\)-indexed diagram, $(S_Σ∘D, S_Σ∘δ)$, by composing with the
  polynominal endofunctor $S_Σ:𝓤^I → 𝓤^I$:
  \begin{equation}
    \begin{aligned}
      (S_Σ∘D)_i     &\defeq S_Σ(D_i) &&(i: \Size)\\
      (S_Σ∘δ)_{i,j} &\defeq S_Σ(δ_{i,j}) &&(i,j: \Size)
    \end{aligned}
  \end{equation}
  Applying $S_Σ$ to the $I$-indexed version of~\eqref{eq:colim-cocone}
  gives a cocone under $(S_Σ∘D, S_Σ∘δ)$ with vertex $S_Σ(\colim D)$ and
  this induces a family of functions in $𝓤^I$
  \begin{equation}
    κ_{D,Σ} : \colim(S_Σ∘D) ⇁  S_Σ(\colim D)
  \end{equation}
  One says that \emph{the polynomial endofunctor $S_Σ:𝓤^I → 𝓤^I$
    preserves $\Size$-indexed colimits} if $κ_{D,Σ}$ is a family of
  isomorphisms, for all diagrams $(D,δ)$. Then the indexed
  generalization of \Cref{cor:poly-cocontinuity} is:
  \begin{quote}
    \emph{In a topos with natural numbers object, given an indexed signature
      $Σ=(I,A,B)$ in a universe $𝓤$ satisfying \IWISC{}, there exists a
      type of sizes $\Size$ with the property that the associated
      polynominal endofunctor $S_Σ:𝓤^I→ 𝓤^I$ preserves $\Size$-indexed
      colimits.}
  \end{quote}
  This is proved as a corollary of the indexed version of
  \cref{thm:cocontinuity} given above, and can also be found in the
  accompanying Agda development~\parencite{fiore2021agdacode}.

\end{remark}

\FloatBarrier
\section{Construction of QWI-types}%
\label{sec:construction}

We aim to prove the following theorem about existence of QWI-types~(\Cref{def:qwit}):
\begin{theorem}%
  \label{thm:main}
  Suppose $𝓤$ is a universe satisfying \IWISC{} in a topos with
  natural numbers object. Then for every indexed signature and system
  of equations in $𝓤$, there exists a QWI-type for it.
\end{theorem}
The proof follows from the cocontinuity results of the previous
section (the indexed versions of \Cref{thm:cocontinuity} and
\Cref{cor:poly-cocontinuity} given in
\cref{rem:indexed-cocontinuity}). For simplicity, here we only give
the proof for QW-types (\cref{sec:qwt}), that is, for the non-indexed
$I=𝟙$ case of signatures. The general case is similar, but
notationally more involved since there are indexes ranging both over
an index type and over sizes. The proof for the general, indexed case
can be found in our Agda development~\parencite{fiore2021agdacode}.

\emph{So in this section we fix a signature $Σ=(A:𝓤,B:A→𝓤)$ and system of
equations $ε=(E:𝓤,V:E→𝓤, l,r: ∏_{e:E}T_Σ(V\,e))$ over it in some
universe $𝓤$ satisfying \IWISC{} in a topos with natural numbers
object.}

\subsection{Motivating the construction}\label{motivatingtheconstruction}

We noted at the start of \Cref{sec:wisc} that QW-types can be
constructed in the category of ZFC sets as initial algebras for
possibly infinitary equational theories by first forming the W-type of
terms of the theory and then quotienting that by the congruence
relation generated by the equations of the theory. AC is used when
constructing the algebra structure of the quotient, because the
signature's arities may be infinite. To avoid this use of AC, instead
of forming all terms in one go and then quotienting, we consider
interleaving quotient formation with the construction of terms of free
algebras for equational systems (cf.~the categorical construction by
\textcite{fiore2009construction}).

\begin{figure}
  \begin{mdframed}
    \[
      \begin{array}{l}
        \mathtt{mutual}\\
        \quad\mathtt{data}\; W : 𝓤\;\mathtt{where}\\
        \quad\quad qτ : T_Σ(W/{∼}) → W\\
        \quad\mathtt{data}\;\_∼\_: W→W→\Prop\;\mathtt{where}\\
        \quad\quad qε : ∀(e : E).∀(ρ : V\; e → W/{∼}).\;
        qτ(T_Σ\,ρ\,(l\,e)) ∼ qτ(T_Σ\,ρ\,(r\,e))\\
        \quad\quad qη : ∀(t:T_Σ(W/{∼})).\; qτ(η\,[qτ\,t]_{∼}) ∼ qτ\,t\\
        \quad\quad qσ : ∀(a:A).∀(b:B\,a→ T_Σ(W/{∼})).\;
        qτ(σ(a,b)) ∼ qτ(σ(a,η∘[\_]_{∼}∘qτ∘b))\\
        \QW = W/{∼}
      \end{array}
    \]
  \end{mdframed}
  \caption{First attempt at constructing QW-types}%
  \label{fig:construction-agda}
\end{figure}

\Cref{fig:construction-agda} gives the idea, using the constructors
$η$ and $σ$ from~\eqref{eq:Tconstructors} and Agda-like notation for
inductive definitions. We would like to construct the QW-type for
$(Σ,ε)$ as a quotient $\QW\defeq W/{∼}$, but now the type $W : 𝓤$ and
the relation $\_∼\_:W→W→\Prop$ are mutually inductively defined, with
constructors as indicated in the figure. Note that whereas the
construction in ZFC uses AC to get an $S_Σ$-algebra structure for
$\QW$, here we get one trivially from the constructor
$qτ:T_Σ(W/{∼})→ W$:
\begin{equation}
  S_Σ(\QW) ≡ S_Σ (W/{∼}) \xrightarrow{S_Σ\,η} S_Σ(T_Σ(W/{∼}))
  \xrightarrow{σ} T_Σ(W/{∼}) \xrightarrow{qτ} W \xrightarrow{[\_]_{∼}}
  W/{∼} ≡ \QW
\end{equation}
Furthermore the property $qε$ of $∼$ in \Cref{fig:construction-agda}
ensures that the $S_Σ$-algebra $W/{∼}$ satisfies the equational system
$ε$. The use of $T_Σ$ rather than $S_Σ$ in the domain of $qτ$ seems
necessary for this method of construction to go through (once we have
fixed up the problems mentioned in the next paragraph); but it does
mean that as well as $qε$, we have to impose the conditions $qη$ and
$qσ$ to ensure that $W/{∼}$ has a $S_Σ$-algebra structure, or
equivalently, an algebra structure for the monad $T_Σ$.

Note that the domain of the constructor $qτ$ combines $T_Σ(\_)$ with
$\_/\_$. While the first is unproblematic for inductive definitions,
the second is not: if one thinks of the semantics of inductively
defined types in terms of initial algebras of endofunctors, it is not
clear what endofunctor (in some class known to have initial algebras)
is involved here, given that both arguments to $\_/\_$ are being
defined simultaneously. Agda uses a notion of ``strict positivity'' as
a conservative approximation for such a class of functors; and one can
instruct Agda to regard quotienting as a strictly positive operation
through the use of its \texttt{POLARITY} declarations. If one does so,
then a definition like the one in \Cref{fig:construction-agda} is
accepted by Agda. The semantic justification for regarding quotients
as strictly positive constructions needs further investigation. We
avoid the need for that here and replace the attempt to define $W$ and
$∼$ inductively by a size-indexed version that uses definition by
well-founded recursion over a type of sizes as in the previous
section. This also avoids another difficulty with
\Cref{fig:construction-agda}: even if one can define $W$ and $∼$
inductively, one still has to verify that $W/{∼}$ has the universal
property~\eqref{def_qwrec}--\eqref{def_qwuniq} required of a
QW-type. In particular, there is an obvious recursive definition of
$\qwrec$ following the shape of the inductive definition in
\Cref{fig:construction-agda}, but it is not at all clear why this
recursive definition is terminating (i.e.~gives a well-defined, total
function). Well-founded recursion over sizes will solve this problem
as well.

\subsection{QW-type via sizes}%
\label{sec:qwi-sized}

We are given a signature $Σ=(A:𝓤,B:A→𝓤)$ and system of equations
$ε=(E:𝓤,V:E→𝓤, l,r: ∏_{e:E}T_Σ(V\,e))$.  Let $\Size:𝓤$ be the type of
sizes whose existence is guaranteed by \Cref{thm:cocontinuity} when in
the theorem we take $A$ to be $AE \defeq A + E$ and $B$
to be the function $BV:AE→𝓤$ mapping $a:A$ to $B\,a$ and $e:E$ to
$V\,e$. It follows (as in the proof of \Cref{cor:poly-cocontinuity})
that $S_Σ$ preserves $\Size$-indexed colimits.

\begin{definition}%
  \label{def:sized-alg}
  A \emph{$\Size$-indexed $Σ$-structure} in $𝓤$ is specified by a family
  of types $D:\Size→ 𝓤$ equipped with functions
  \begin{equation}
    τ_{j,i}:T_Σ(D_j) → D_i \quad\text{for all $i,j:\Size$ with $j<i$}
  \end{equation}
  Similarly, for each $i:\Size$, a \emph{$(↓ i)$-indexed
  $Σ$-structure} is the same thing, except with $D$ only defined on
  the subsemicategory $↓ i$~\eqref{eq:down-seg} rather than the whole
  of $\Size$. Clearly, given a $\Size$-indexed $Σ$-structure $(D,τ)$,
  for each $i:\Size$ we get a $(↓ i)$-indexed one $(D\;(↓ i), τ\;(↓
  i))$ by restriction.
\end{definition}
If $i:\Size$ and $(D,τ)$ is a $(↓ i)$-indexed $Σ$-structure, let $◇_i D :
𝓤$ be the quotient type (see \cref{fig:quot})
\begin{equation}%
  \label{eq:diamond}
  ◇_i D\;\defeq\;
  \left.\left(∑_{j<i} T_Σ\,D_j\right)\right/R_i
\end{equation}
with the relation $R_i: (∑_{j<i}T_Σ\,D_j)→(∑_{j<i} T_Σ\,D_j)→\Prop$
defined by:
\begin{multline}%
  \label{eq:diamond-rel}
  R_i\,(j,t)\,(k,t') \;\defeq\; {}\\
  \begin{aligned}[t]
    & &&({k=j} \;∧\; ∃(e:E).∃(ρ:V\,e→D_j).\;{t=T_Σ\,ρ\,(l\,e)} \;∧\;
    {t'=T_Σ\,ρ\,(r\,e)})\\
    &∨ &&({k<j} \;∧\; t = η(τ_{k,j}\,t'))\\
    &∨ &&({k<j} \;∧\; ∃(a:A).∃(b:B\,a→ T_ΣD_k).\; {t= σ(a,b)} \;∧\;
    {t'=σ(a,η∘τ_{k,j}∘b)})
  \end{aligned}
\end{multline}
(The three clauses in the above definition correspond to the three
constructors $qε$, $qη$ and $qσ$ in \Cref{fig:construction-agda}.)  We
will see that the QW-type for $(Σ,ε)$ can be constructed from
$\Size$-indexed $Σ$-structures $(D,τ)$ satisfying the following
fixed-point property.

\begin{definition}%
  \label{def:fixed}
  A $\Size$-indexed $Σ$-structure $(D,τ)$ is a \emph{$◇$-fixed point}
  if for all $i:\Size$
  \begin{equation}
    D_i = ◇_i(D\;(↓ i))\label{eq:fixed1}
  \end{equation}
  and for all $j<i$ and $t:  T_ΣD_j$
  \begin{equation}
    τ_{j,i}\,t = [j,t]_{R_i}\label{eq:fixed2}
  \end{equation}
\end{definition}
Suppose that $(D,τ)$ is a $◇$-fixed point. Because of~\eqref{eq:fixed1} and~\eqref{eq:fixed2}, for $i,j:\Size$ with $i<j$
the functions
\begin{equation}%
  \label{eq:fixed-act}
  δ_{i,j} : D_i → D_j \qquad
  δ_{i,j} \defeq τ_{i,j} ∘ η
\end{equation}
satisfy
\begin{equation}%
  \label{eq:fixed-diag}
  δ_{i,j}([k,t]_{R_i}) = [k,t]_{R_j} \quad\text{(for all $k<i$ and
    $t:T_Σ\,D_k$)}
\end{equation}
In particular they satisfy composition~\eqref{eq:diag} and so $(D,δ)$
is a $\Size$-indexed diagram in $𝓤$ whose colimit we can form as in
\Cref{sec:colimit}.  We prove that this colimit
\begin{equation}%
  \label{eq:qw-type}
  \QW \;\defeq\; \colim D
\end{equation}
has the structure~\eqref{def_qwintro}--\eqref{def_qwuniq} of a
QW-type for the signature $(Σ,ε)$.

\begin{description}
  \setlength\mathindent{\leftmargin+\originalparindent}
\item[$\qwintro$] Given how we chose $\Size$ at the start of this
  section, from the (proof of) \Cref{cor:poly-cocontinuity} we have
  that $S_Σ:𝓤→𝓤$ preserves $\Size$-indexed colimits. We claim that the
  functions%
  \begin{equation}%
    \label{eq:qwintro}
    S_Σ(D_i) \smashxrarrow{S_Ση} S_Σ(T_Σ\,D_i)\smashxrarrow{σ} T_Σ\,D_i
    \smashxrarrow{τ_{i,\Suc{i}}} D_{\Suc{i}}
    \smashxrarrow{(ν_D)_{\Suc{i}}}
    \colim D
  \end{equation}
  form a cocone under the diagram $S_Σ∘D$ and hence induce a function
  $\colim(S_Σ∘D) → \colim D$; then we obtain
  $\qwintro: S_Σ(\colim D)→ \colim D$ by composing this with the
  isomorphism $S_Σ(\colim D) ≅ \colim(S_Σ∘ D)$ from
  \Cref{cor:poly-cocontinuity}. That the functions in~\eqref{eq:qwintro} form a cocone for $S_Σ∘D$ follows from the fact
  that the functions in~\eqref{eq:fixed-act} satisfy for all sizes
  $i<j$
  \begin{equation}%
    \label{eq:qwintro2}
    ∀(t:T_Σ(Dᵢ)).∃(k:\Size).\; j<k \;∧\; τ_{j,k}(T_Σ\,δ_{i,j}(t)) =
    τ_{i,k}(t)
  \end{equation}
  from which it follows that
  $(ν_D)_{\Suc{j}}∘τ_{j, \Suc{j}} ∘ T_Σ\,δ_{i,j} = (ν_D)_{\Suc{i}}∘τ_{i,
  \Suc{i}}$ and hence also the cocone property. Property~\eqref{eq:qwintro2}
  can be proved by induction on the structure of $t:T_Σ(Dᵢ)$, with the
  $t=σ(a,b)$ case of~\eqref{eq:Tconstructors} proved using the fact
  that $\Size$ has $<$-upper bounds for families of sizes indexed by
  $B\,a$ (\Cref{rem:cocontinuity}).

\item[$\qwequate$] Given $e:E$ and $ρ:V\,e → \colim D$, by
  \Cref{thm:cocontinuity} there exist $i:\Size$ and $ρ':V\,e → D_i$
  with $\rho = (ν_D)_i ∘ ρ'$. By standard properties of the bind
  operation $⋙$~\eqref{def_bind} (which depends implicitly on the
  $S_Σ$-algebra structure $\qwintro$ that we have just constructed), we
  have
  \begin{equation}
    \begin{array}{@{}r@{}c@{}r@{}c@{}l@{}r@{}}
      {}  (l\,e ⋙  ρ)
          &{}={}& (l\,e ⋙ (ν_D)_i ∘ ρ')
          &{}={}& (T_Σ\,ρ'\,&(l\,e) ⋙ (ν_D)_i)
      \\  (r\,e ⋙  ρ)
          &{}={}& (r\,e ⋙ (ν_D)_i ∘ ρ')
          &{}={}& (T_Σ\,ρ'\,&(r\,e) ⋙ (ν_D)_i)
    \end{array}
  \end{equation}
  From the definition of $\qwintro$ and~\eqref{eq:fixed2} it follows
  that for any $t:T_Σ\,D_i$, there is a proof of
  \begin{equation}%
    \label{eq:qwequate}
    (t ⋙ (ν_D)_i) = (ν_D)_{\Suc{i}}(τ_{i,\Suc{i}}\,t)
  \end{equation}
  So from above we have that
  $(l\,e ⋙ ρ) \;=\; (ν_D)_{\Suc{i}}(τ_{i,\Suc{i}}(T_Σ ρ'(l \,e)))$ and
  $(r\,e ⋙ ρ) \;=\; (ν_D)_{\Suc{i}}(τ_{i,\Suc{i}}(T_Σ ρ'(r \,e)))$.
  Since from~\eqref{eq:fixed2} and the first clause in the definition
  of $Rᵢ$~\eqref{eq:diamond-rel} we also have
  $τ_{i,\Suc{i}}(T_Σ ρ'(l\,e)) = τ_{i,\Suc{i}}(T_Σ ρ'(r\,e))$, it
  follows that there is a proof of $(l\,e ⋙ ρ) = (r\,e ⋙ ρ)$.

\item[$\qwrec$] Given an $S_Σ$-algebra $(X,α):∑_{X:𝓤}(S_Σ\,X→X)$
  satisfying the system of equations $ε$, the function
  $\qwrec:\colim D → X$ is induced by a cocone of functions
  $r:\prod_i(D_i→ X)$ under the diagram $D$, defined by well-founded
  recursion (\Cref{prop:wfrec}). More precisely, a strengthened
  ``recursion hypothesis'' is needed: instead of $\prod_i(D_i→ X)$ we
  use $∏_i\,F_i$ where
  \begin{equation}
    F_i \;\defeq\; \{f:D_i→ X ∣ ∀j <i.∀(t :T_Σ\,D_j).\; (t ⋙
      (f∘δ_{j,i})) = f ([j,t]_{R_i})\}
  \end{equation}
  (The definition uses $α$ implicitly in $ ⋙$ and relies on the
  fact~\eqref{eq:fixed1} that $D_i = ◇_i(D\;(↓ i))$.) For each
  $i:\Size$, if we have $r_j:F_j$ for all $j<i$, then we get a function
  $r_i:F_i$ well-defined by
  \begin{equation}%
    \label{eq:qwrec}
    r_i([j,t]_{R_i}) \defeq t ⋙ r_j \quad\text{for all $j<i$ and
      $t:T_Σ\,D_j$}
  \end{equation}
  (The defining property of $F_i$ is needed to see that the right-hand
  side of this definition respects the relation $R_i$.) Hence by
  well-founded recursion (\Cref{prop:wfrec}) we get an element
  $r:∏_i\,F_i$. One can prove $∀i.∀j<i.\; r_j = r_i∘δ_{j,i}$ by
  well-founded induction~\eqref{eq:<iswf}; so $r$ is a cocone and
  induces a function $\qwrec:\colim D → X$.

\item[$\qwrechom$] We noted at the start of this section that by
  choice of $\Size$, the functor $S_Σ$ preserves $\Size$-indexed
  colimits.  So to prove that
  \begin{equation}
    \xymatrix{%
      {S_Σ(\colim D)} \ar[r]^<<<<<{\qwintro} \ar[d]_{S_Σ\qwrec}
      & {\colim D} \ar[d]^{\qwrec}\\
      {S_Σ\,X}\ar[r]_{α} &X}
  \end{equation}
  commutes, by the definitions of $\qwintro$ and
  $\qwrec$, it suffices to prove that
  \begin{equation}
    \xymatrix{%
      {S_Σ\,D_i} \ar[rr]^{τ_{i,\Suc{i}}∘σ∘S_Σ η} \ar[d]_{S_Σ\,r_i}
      && {D_{\Suc{i}}} \ar[d]^{r_{\Suc{i}}}\\
      {S_Σ\,X}\ar[rr]_{α} &&X}
  \end{equation}
  does for each $i:\Size$. But each
  $(a,b): S_Σ\,D_i$ is mapped by $τ_{i,\Suc{i}}∘σ∘S_Σ η$ to
  $[i,σ(a,η∘b)]_{R_{\Suc{i}}}$ (using the fact that
  $D_{\Suc{i}} = ◇_{\Suc{i}}(D\;(↓ \Suc{i}))$); and by~\eqref{eq:qwrec},
  that is mapped by $r_{\Suc{i}}$ to $σ(a,η∘b) ⋙ r_i$, which is indeed
  equal to $α(S_Σ\,r_i(a,b))$ by definition of $⋙$~\eqref{def_bind}
  and the action of $S_Σ$ on
  functions~\eqref{def_S_endofunctor-action}.

\item[$\qwuniq$] If $h : \colim D → X$ is a morphism of
  $S_Σ$-algebras, then one can prove by well-founded induction for $<$
  that $∀ i.\; h∘(\nu_D)_i = r_i$ holds: for if we have $h∘(\nu_D)_j =
  r_j$ for all $j<i$, then for any $[j,t]_{R_i}$ in $D_i = ◇_i(D\;(↓ i))$
  \[
    \begin{array}{r@{}c@{}l@{\hspace{2em}}l}
      {}  r_i([j,t]_{R_i})
          &{}≝{}& t ⋙ r_j
          &\text{\eqref{eq:qwrec}}
      \\  &=& t ⋙ (h∘(νD)_j)
          &\text{by induction hypothesis}
      \\  &=& h(t ⋙ (\nu_D)_j)
          &\text{since $h$ is a morphism of $S_Σ$-algebras}
      \\  &=& h((ν_D)_{\Suc{j}}[j,t]_{R_{\Suc{j}}})
          &\text{by~\eqref{eq:qwequate} and~\eqref{eq:fixed2}}
      \\  &=& h((ν_D)_i[j,t]_{R_i})
          &\parbox[c]{0.43\textwidth}{\raggedright{}since by~\eqref{eq:fixed-diag}, $δ_{\Suc{j}, \Suc{j} ⊔ˢ
            i}([j,t]_{R_{\Suc{j}}}) = [j, t]_{R_{\Suc{j} ⊔ˢ i}} = δ_{i,\Suc{j} 
            ⊔ˢ i}([j,t]_{R_i})$.} 
    \end{array}
  \]
  So by well-founded induction, $h∘(\nu_D)_i = r_i$ holds for all
  $i:\Size$, and hence by definition of $\qwrec$ and the uniqueness
  part of the universal property of colimits we have $h=\qwrec$.

\end{description}
Thus we have proved:
\begin{proposition}%
  \label{prop:existence}
  Let $\Size$ be the type of sizes defined at the start of this
  section, whose existence is guaranteed by \Cref{thm:cocontinuity}.
  If the $\Size$-indexed $Σ$-structure $(D,τ)$ is a $◇$-fixed point,
  then $\colim D$ has the structure of a QW-type for the signature
  $(Σ,ε)$. \qed%
\end{proposition}
Now we can complete the proof of the main theorem:

\begin{proof}[Proof of non-indexed version of \Cref{thm:main}]
  In view of \Cref{prop:existence}, it suffices to construct a
  $\Size$-indexed $Σ$-structure which is a $◇$-fixed point in the sense
  of \Cref{def:fixed}.

  For each $i:\Size$, say that a $(↓ i)$-indexed $Σ$-structure
  (\Cref{def:sized-alg}) is an \emph{upto-$i$} $◇$-fixed point if
  \begin{equation}%
    \label{eq:upto}
    ∀j<i.\; D_j=◇_j(D\;(↓ j)) \;∧\; ∀ k<j.∀ t.\; τ_{k,j}\,t \heq [k,t]_{R_j}
  \end{equation}
  (cf.~\eqref{eq:fixed1} and~\eqref{eq:fixed2}).
  Note that:

  \begin{enumerate}[ref={(\Alph*)}, label={(\Alph*)}, widest=A]
    \item\label{iToJRestr} \emph{Given $j<i$, any upto-$i$ $◇$-fixed point
      restricts to an upto-$j$ $◇$-fixed point.}

    \item\label{fixedPointEq} \emph{For all $i$, any two upto-$i$ $◇$-fixed
      points are equal} (proof by well-founded induction~\eqref{eq:<iswf}).
  \end{enumerate}
  Using these two facts, it follows by well-founded recursion
  (\Cref{prop:wfrec}) that there is an upto-$i$ $◇$-fixed point for
  all $i:\Size$. For if $(D^{(j)}, τ^{(j)})$ is an upto-$j$ $◇$-fixed
  point for all $j<i$, then we get $D^{(i)} : {↓ i} → 𝓤$ by defining
  for each $j<i$
  \begin{equation}
    (D^{(i)})_j \defeq ◇_j(D^{(j)})
  \end{equation}
  If $k<j<i$, then by \cref{iToJRestr} we have that $D^{(j)}(↓ k)$ is an upto-$k$
  $◇$-fixed point and hence by \cref{fixedPointEq} that $D^{(j)}(↓ k) = D^{(k)}$.
  So together with~\eqref{eq:upto} this gives:
  \begin{equation}
    (D^{(i)})_k \defeq ◇_k(D^{(k)}) = ◇_k(D^{(j)}(↓ k)) = (D^{(j)})_k
  \end{equation}
  Hence we can define
  $(τ^{(i)})_{k,j} : T_Σ((D^{(i)})_k)→ (D^{(i)})_j$ by
  $(τ^{(i)})_{k,j}\,t \defeq [k, t]_{R_j}$ and this makes
  $(D^{(i)},τ^{(i)})$ into a $(↓ i)$-indexed $Σ$-structure which by
  construction is an upto-$i$ $◇$-fixed point. Thus by well-founded
  recursion (\Cref{prop:wfrec}) we have an upto-$i$ $◇$-fixed point
  $D^{(i)}$ for all $i:\Size$.

  Let $D:\Size→ 𝓤$ be given by $D_i \defeq ◇_i(D^{(i)})$. If $j<i$, then by
  \cref{iToJRestr,fixedPointEq} we have $D^{(i)}(↓ j)  = D^{(j)}$ and together
  with~\eqref{eq:upto} this gives:
  \begin{equation}
    D_j \defeq ◇_j(D^{(j)}) = ◇_j(D^{(i)}(↓ j)) = (D^{(i)})_j
  \end{equation}
  So we can define $τ_{j,i}:T_Σ\,D_j → D_i$ by
  $τ_{j,i}\,t \defeq [j, t]_{R_i}$. This makes $(D,τ)$ into a $Σ$-structure
  which by construction is a $◇$-fixed point.
\end{proof}


\FloatBarrier
\section{Encoding QITs as QWI-types}%
\label{sec:qit_encoding}

The general notion of indexed quotient inductive type (QIT) was
discussed by example in the Introduction. We wish to show that a wide
variety of QITs can be expressed as QWI-types, namely those which do
not use conditional equality constructors (as in the example in~\eqref{alt-bag_qit}). We first introduce a schema for such QITs that
combines desirable features of ones that occur in the literature. As
in the previous section, for simplicity sake we will confine our
attention to the non-indexed case of QITs and QW-types.

\subsection{General QIT schemas and encodings}%
\label{sec:general-qit-schemas}

\begin{emergency}{3em}
\Textcite{basold2017higher} present a schema for infinitary QITs that
do not support conditional path equations. Constructors are defined by
arbitrary polynomial endofunctors built up using (non-dependent)
products and sums, which means in particular that parameters and
arguments can occur in any order; however, they require constructors
to be in uncurried form. They also construct a model of simple 1-cell
complexes and other non-recursive QITs.
\end{emergency}

\Textcite[Sections~3.1 and 3.2]{dybjer2018finitary} present a schema
for finitary QITs that does allow curried constructors (and also
supports \emph{conditional} path equations), but requires all
parameters to appear before all arguments. This contrasts with the
more convenient schema for regular inductive types in Agda, which
allows parameters and arguments in any order.

\Textcite{kaposi2019constructing} define an encoding of finitary
quotient inductive-inductive (QIIT) types, called the \emph{theory of
  signatures} (ToS), which is a restriction of the \emph{theory of
  codes} for HIITs \parencite{kaposi2018syntax}. In the ToS a QI(I)T
is encoded as a context of a small internal type theory. This encoding
is not quite as convenient as the schema for regular inductive types
in Agda, but, when using named variables, it is much closer to Agda
than QWI-types are. They also construct a model for finitary QIITs
(and \textcite{KaposiA:lariqi} reduce the problem of modelling
infinitary QIITs to just one “universal” QIIT – a generalised ToS).

Building on the first two schemas mentioned above, we provide a schema
for infinitary, non-conditional QITs combining the arbitrarily ordered
parameters and arguments of the first \parencite{basold2017higher}
with the curried constructors of the second
\parencite{dybjer2018finitary}.

\begin{figure}[!htbp]
  \small
  \newcommand{\myvspace}{\vspace{0.75\baselineskip}}
  \begin{mdframed}
    \centering

    \begin{minipage}{\linewidth}
      In the following we fix the context \(Γ\) in which the QIT,
      \(Q\), (\(Q ∉ Γ\)) is defined, whereas \(Δ\), \(A\), \(B\),
      \(H\), \(K\), \(a\), \(b\), \(x\), \(y\), etc. are metavariables.
      When deriving an element or equality constructor, Δ will always
      be equal to Γ or an extension of it.
    \end{minipage}

    \myvspace

    \begin{minipage}{\linewidth}
      A type \emph{strictly-positive in \(Q\)} is made up with ∏-types
      and ∑-types and can use \(Q\) and constants or variables in the
      context, provided that \(Q\) never occurs on the LHS of a
      ∏-type, and that the RHS of a ∑-type never depends on \(Q\),
      even if it occurs on the LHS.
    \end{minipage}

    \myvspace

    \fbox{\(Δ ⊢ K \StrPstv\)}

    \myvspace

    \hfill
    \begin{prooftree}
      \hypo{Δ ⊢ B : 𝓤}
      \infer1[\ConstantParameter]{Δ ⊢ B \StrPstv}
    \end{prooftree}
    \hfill
    \begin{prooftree}
      \hypo{Δ ⊢ B : 𝓤}
      \hypo{Δ, b : B ⊢ K \StrPstv}
      \infer2[\StrictlyPositiveFunction]{Δ ⊢ ∏_{b : B} K \StrPstv}
    \end{prooftree}
    \hfill\vphantom{.}

    \myvspace

    \hfill
    \begin{prooftree}
      \infer0[\InductiveArgument]{Δ ⊢ Q \StrPstv}
    \end{prooftree}
    \hfill
    \begin{prooftree}
      \hypo{Δ ⊢ A \StrPstv}
      \hypo{Δ, a : A\subst{𝟙}{Q} ⊢ K \StrPstv}
      \infer2[\StrictlyPositiveProduct]{Δ ⊢ ∑_{a : A} K' \StrPstv}
    \end{prooftree}
    \hfill\vphantom{.}

    \myvspace

    \begin{minipage}{\linewidth}
      Where \(K'\) is found from \(K\) by recursively replacing each
      sub-term of type \(𝟙\), \(A → 𝟙\), etc. with the corresponding
      unique term \(0\), \(!\), etc.
    \end{minipage}

    \myvspace

    \begin{minipage}{\linewidth}
      \emph{%
        Note: We can see that these rules give strictly-positive types
        because: (1) at this point \(Δ ⊢ Q : 𝓤\) is not derivable
        since \(Q ∉ Γ\) and no rule for \(\StrPstv\) adds \(Q\) to the
        context, so \(Q\) can never appear in the LHS of a ∏-type.
        Also (2) the \(\StrictlyPositiveProduct\) rule ensures that
        abstracting with \(∑\) cannot result in dependencies of \(Q\)
        in the codomain.
      }
    \end{minipage}

    \myvspace

    \begin{minipage}{\linewidth}
      The type of an element constructor is an (iterated) ∏-type with
      codomain \(Q\), with zero or more strictly-positive arguments.
    \end{minipage}

    \myvspace

    \fbox{\(Δ ⊢ H \ElCnstr\)}

    \myvspace

    \hfill
    \begin{prooftree}
      \hypo{Δ, Q : 𝓤 ⊢ H \ElType}
      \infer1[\ElCon]{Δ ⊢ H \ElCnstr}
    \end{prooftree}
    \hfill\vphantom{.}

    \myvspace

    \hfill
    \begin{prooftree}
      \infer0[\Target]{Δ ⊢ Q \ElType}
    \end{prooftree}
    \hfill
    \begin{prooftree}
      \hypo{Δ\subst{𝟙}{Q} ⊢ K \StrPstv}
      \hypo{Δ, a : K' ⊢ H \ElType}
      \infer2[\ElArgumento]{Δ ⊢ ∏_{a : K'} H \ElType}
    \end{prooftree}
    \hfill\vphantom{.}

    \myvspace

    \begin{minipage}{\linewidth}
      Where \(Δ\subst{𝟙}{Q}\) is the context that results after
      replacing occurrences of \(Q\) in \(Δ\) by the type \(𝟙\).
    \end{minipage}

    \myvspace

    \begin{minipage}{\linewidth}
      \emph{%
        Note: Deriving \(Γ ⊢ ∏_{a : K'} Q \ElCnstr\) via
        \ElArgumento{} for some strictly-positive \(K\) does not imply
        \(Γ ⊢ ∏_{a : K'} Q : 𝓤\) since \(Γ ⊬ Q : 𝓤\) (since \(Q ∉
        Γ\)), and also \(Γ ⊬ K' : 𝓤\) in general since \(K'\) may
        contain \(Q\). This distinction also applies to the following
        judgements.
      }
    \end{minipage}

    \myvspace

    \begin{minipage}{\linewidth}
      The type of an equality constructor is an (iterated) ∏-type with
      codomain \(x =_Q y\), with zero or more strictly-positive
      arguments. Note that each of the element constructors are added
      to the context and can be used to derive \(x : Q\) and \(y :
      Q\).
    \end{minipage}

    \myvspace

    \fbox{\(Δ ⊢ K \EqCnstr\)}

    \myvspace

    \begin{prooftree}
      \hypo{Δ, Q : 𝓤, c₁ : C₁, ⋯, cₘ : Cₘ ⊢ K : \EqType}
      \infer1[\EqCon]{Δ ⊢ K \EqCnstr}
    \end{prooftree}

    \myvspace

    \begin{prooftree}
      \hypo{Δ ⊢ x : Q}
      \hypo{Δ ⊢ y : Q}
      \infer2[\EqTarget]{Δ ⊢ x = y \EqType}
    \end{prooftree}
    \hfill
    \begin{prooftree}
      \hypo{Δ\subst{𝟙}{Q} ⊢ K \StrPstv}
      \hypo{Δ, a : K' ⊢ H \EqType}
      \infer2[\EqArg]{Δ ⊢ ∏_{a : K'} H \EqType}
    \end{prooftree}

  \end{mdframed}

  \caption{Rules for QIT element and equality constructors.}%
  \label{qit-constructor-types}
\end{figure}

The following definition should be read as an extension of a
formalisation of the type theory described in \Cref{sec:type-theory}
in terms of typing contexts ($Γ,Δ,\ldots$) and various
judgements-in-context (such as $Γ ⊢ a : A$, $Γ ⊢ a = a' : A$, etc.);
see the HoTT Book~\parencite[Appendix]{HoTT}.

\begin{definition}\label{def_equational_QITs}
  \begin{emergency}{10em}
    A (\emph{non-conditional}) \emph{QIT}, \(Q\), in a context \(Γ\),
    \(Q ∉ Γ\) is specified by a list of element constructors \({c₁ :
    C₁,} …, {cₘ : Cₘ}\), followed by a list of equality constructors
    \({d₁ : D₁,} …, {dₙ : Dₙ}\), where the \(Cᵢ\)s and \(Dⱼ\)s are
    derived as \(Γ ⊢ Cᵢ\,\ElCnstr\) and \(Γ ⊢ Dⱼ\,\EqCnstr\)
    respectively, according to the rules in
    \cref{qit-constructor-types}.
  \end{emergency}

  Given a list of element and equality constructors built up according to these
  rules, the newly-defined QIT, \(Q\), has formation rule
  \[
    \begin{prooftree}
      \infer0{Γ ⊢ Q : 𝓤}
    \end{prooftree}
  \]
  for the specific context \(Γ\) in which it was defined. And it
  has an introduction rule for each element constructor and each equality
  constructor:
  \[
    \begin{prooftree}
      \infer0{Γ ⊢ c₁ : C₁}
    \end{prooftree}
    \quad\cdots\quad
    \begin{prooftree}
      \infer0{Γ ⊢ cₘ : Cₘ}
    \end{prooftree}
    \qquad\qquad
    \begin{prooftree}
      \infer0{Γ ⊢ d₁ : D₁}
    \end{prooftree}
    \quad\cdots\quad
    \begin{prooftree}
      \infer0{Γ ⊢ dₙ : Dₙ}
    \end{prooftree}
  \]
\end{definition}
We show how to derive elimination and computation rules for $Q$ from
these formation and introduction rules (compare with
\textcite[Definition~10]{basold2017higher}) after first giving an
example of how the schema is instantiated.

\begin{example}
  Consider multisets as in example~\eqref{bag_qit} with two element
  constructors \([]\) and \(\_∷\_\), and one equality constructor
  \(\swap\). Given the context containing parameter \(X : 𝓤\), we
  define the QIT \(·, X \mathbin: 𝓤 ⊢ \Bag \mathbin: 𝓤₁\) by
  providing those constructors along with their types such that they
  are derived by the rules in \cref{qit-constructor-types}:

  Let \(Γ ≝ (·, X : 𝓤)\) and \(Δ ≝ (Γ, \Bag : 𝓤) = (·, X : 𝓤, \Bag :
  𝓤)\); so that, by definition, \(Δ\subst{𝟙}{\Bag} = Γ\).
  \begin{itemize}
    \item \([]\)
      {\tiny
        \begin{center}
          \begin{prooftree}
            \infer0[\Target]{Δ ⊢ \Bag \ElType}
            \infer1[\ElCon]{Γ ⊢ \Bag \ElCnstr}
          \end{prooftree}
        \end{center}
      }

    \item \(\_∷\_\)
      {\tiny
        \begin{center}
          \begin{prooftree}
            \infer0{Γ ⊢ X : 𝓤}
            \infer1[\ConstantParameter]{Γ ⊢ X \StrPstv}
            \infer0[\InductiveArgument]{Γ, a : X ⊢ \Bag \StrPstv}
            \infer0[\Target]{Δ, a : X, b : \Bag ⊢ \Bag \ElType}
            \infer2[\ElArgumento]{Δ, a : X ⊢ (\Bag → \Bag) \ElType}
            \infer2[\ElArgumento]{Δ ⊢ (X → \Bag → \Bag) \ElType}
            \infer1[\ElCon]{Γ ⊢ (X → \Bag → \Bag) \ElCnstr}
          \end{prooftree}
        \end{center}
      }
      \vspace{0.5\baselineskip}
  \end{itemize}
  \newcommand{\eqn}{\mathit{eqn}}
  \hspace{\parindent}Let
  \begin{align*}
    {}  Θ &≝ Δ, [] : \Bag, \_∷\_ : X → \Bag → \Bag
    \\  Ξ &≝ Θ\subst{𝟙}{\Bag} ≝ Γ, [] : 𝟙, \_∷\_ : X → 𝟙 → 𝟙
    \\  Φ &≝ Θ, x\;y : X, \zs : \Bag
    \\  \eqn &≝ x ∷ y ∷ \zs = y ∷ x ∷ \zs
  \end{align*}
  \begin{samepage}
  \begin{itemize}
    \item \swap
  \end{itemize}
  {\tiny
    \begin{center}
      \begin{prooftree}
        \infer0{Ξ ⊢ X : 𝓤}
        \infer1{Ξ ⊢ X \StrPstv}
        \infer0{Ξ, x : X ⊢ X : 𝓤}
        \infer1{Ξ, x : X ⊢ X \StrPstv}
        \infer0{Ξ, x\;y : X ⊢ \Bag \StrPstv}
        \hypo{⋮}
        \infer1{Φ ⊢ x ∷ y ∷ \zs : \Bag}
        \hypo{⋮}
        \infer1{Φ ⊢ y ∷ x ∷ \zs : \Bag}
        \infer2[\EqTarget]{
          Φ ⊢ \eqn \EqType
        }
        \infer2[\EqArg]{
          Θ, x\;y : X ⊢ ∏_{\zs : \Bag}
            \eqn \EqType
        }
        \infer2[\EqArg]{
          Θ, x : X ⊢ ∏_{y : X} ∏_{\zs : \Bag}
            \eqn \EqType
        }
        \infer2[\EqArg]{
          Θ ⊢ ∏_{x\;y : X} ∏_{\zs : \Bag}
            \eqn \EqType
        }
        \infer1[\EqCon]{Γ ⊢ ∏{x\;y : X} ∏_{\zs : \Bag} \eqn \EqCnstr}
      \end{prooftree}
    \end{center}
  }
  \end{samepage}
\end{example}

\begin{definition}[Elimination and computation]
  The \emph{arguments} of a constructor \(Cⱼ\), written \(\Arg(Cⱼ)\),
  is the list of \(K'\) for all strictly-positive types \(K\)
  introduced with the rule \ElArgumento.

  Given a QIT defined by constructors \(c₁,…,cₘ,d₁,…,dₙ\) as above,
  the \emph{underlying inductive type} \(⌊Q⌋\) is the inductive type
  defined by only the element constructors \(c₁,…,cₘ\), ignoring the
  equalities; then the \emph{underlying W-type} is the W-type that
  encodes this inductive type. (Recall that every strictly-positive
  description of a polynomial endofunctor gives rise to a W-type with
  the same initial algebra; see \textcite{dybjer1997representing}, and
  for the more general indexed and nested case see
  \textcite{altenkirch2015indexed}.)

  Define \emph{leaf application} on a strictly-positive term \(a : A\)
  (that is, a term with a strictly-positive type) for a function \(f :
  ∏_{q : Q} R\,q\) into some type \(R\), written \(f \leafapp a\), by
  induction on the \StrPstv{} structure of \(A\):
  \begin{equation}
    \begin{aligned}
      {}& \InductiveArgument & f \leafapp x &≝ f\,x
      \\& \ConstantParameter & f \leafapp b &≝ b
      \\& \StrictlyPositiveProduct & f \leafapp (a, b) &≝ (f \leafapp a, f \leafapp b)
      \\& \StrictlyPositiveFunction & f \leafapp g &≝ (f \leafapp\_) ∘ g
    \end{aligned}
  \end{equation}
  In order to define the elimination and computation rules we must
  first define the induction step. First, given a ``motive''
  \(Γ ⊢ P : Q → 𝓤\), define the type \(P' ≝ ∑_{x : Q} P\,x\).
  Now given a strictly-positive argument \(A\), define the type \(A'\)
  (it will be (one of) the induction hypotheses) by induction on the
  structure of \(A\), replacing each occurrence of \(Q\)
  (\InductiveArgument) with \(P'\). Terms \(a' : A'\) are also
  strictly-positive and admit leaf application for functions of type
  \(∏_{p' : P'} S\,p'\) for some type \(S\).

  To find the induction step \(\wat{Cⱼ}\) for each element constructor
  \(cⱼ : Cⱼ\), replace each of the strictly-positive arguments \(∏_{a₁
  :A₁}\,⋯\,∏_{aₚ : Aₚ}\) with \(∏_{a₁ :A₁'}\,⋯\,∏_{aₚ : Aₚ'}\) , as
  defined above, and replace the target \(Q\) of the constructor with
  \(P\,(cⱼ\,(π₁ \leafapp a₁)\,⋯\,(π₁ \leafapp aₚ))\).

  The induction cases for each of the element constructors are then \(h₁
  : \wat{C₁}, …, hₘ : \wat{Cₘ}\), and these provide an eliminator
  \(\elim\,h₁\,⋯\,hₘ\) for the underlying inductive type \(⌊Q⌋\).

  Given \(h₁,…,hₘ\), define \(\wat{Dₖ}\), for each equality
  constructor \(dₖ : Dₖ\), in the same way. Replace the homogeneous
  equality type with a heterogeneous equality type and replace the
  endpoints \(l,r\) of the equality with \(\wat{l},\wat{r}\)
  inductively defined by cases depending on if the abstracted variable
  is a constructor or not:
  \begin{equation}
    \begin{aligned}
      {}\wat{cⱼ} &≝ hⱼ
      \\\wat{x} &≝ x : P'ᵢ
    \end{aligned}
  \end{equation}
  \begin{samepage}
  The elimination rule is then:
  \begin{center}
    \begin{prooftree}
      \hypo{Γ ⊢ P : Q → 𝓤}
      \infer[no rule]1{
        {}      Γ ⊢ h₁ : \wat{C₁}
        \quad   …
        \quad   Γ ⊢ hₘ : \wat{Cₘ}
      }
      \infer[no rule]1{
        {}      Γ ⊢ p₁ : \wat{D₁}
        \quad   …
        \quad   Γ ⊢ p_{\mathrlap{n}\phantom{m}}
        : \wat{D_{\mathrlap{n}\phantom{m}}}
      }
      \infer1[\QITelim]{\qitelim\,h₁\,⋯\,hₘ\,p₁\,⋯\,pₙ : ∏_{x : Q}
        P\,x}
    \end{prooftree}
    \vspace{0.4\baselineskip}
  \end{center}
  \end{samepage}
  Finally the computation rules are, for each element constructor \(cᵢ
  : Cᵢ\),
  \begin{center}
    \vspace{0.4\baselineskip}
    \begin{prooftree}
      \hypo{\text{\emph{same hypotheses as \QITelim}}}
      \hypo{Γ ⊢ a₁,…,aₚ : \Arg(Cᵢ)}
      \infer2[\QITcomp]{
        \qitelim\,h₁\,⋯\,hₘ\,p₁\,⋯\,pₙ
        \,(cᵢ\,a₁\,⋯\,aₚ)
      }
      \infer[no rule]1{
        =_{P\,(cᵢ\,a₁\,⋯\,aₚ)} hᵢ
        \,(a₁ , (\qitelim\,⋯) \leafapp a₁)
        \,⋯
        \,(aₚ , (\qitelim\,⋯) \leafapp aₚ)
      }
    \end{prooftree}
  \end{center}
\end{definition}

\begin{samepage}
\begin{example}
  \newcommand{\nil}{\mathit{nil}}
  \newcommand{\cons}{\mathit{cons}}
  \renewcommand{\resp}{\mathit{resp}}
  The elimination rule for \(X : 𝓤 ⊢ \Bag\) is:
  \begin{center}
    \footnotesize
    \vspace{0.4\baselineskip}
    \begin{prooftree}
      \hypo{
        {} Γ ⊢ P : \Bag → 𝓤
        \qquad  Γ ⊢ \nil : P\,[]
        \qquad  Γ ⊢ \cons : ∏_{x : X}
          ∏_{\xs' : \left(∑_{\xs : \Bag} P\,\xs\right)}
          P\,(x ∷ (π₁\;\xs'))
      }
      \infer[no rule]1{
        Γ ⊢ \resp :
          ∏_{x\;y : X} ∏_{\zs' : \left(∑_{\zs : \Bag} P\,\zs\right)}
          \cons\;x\;(y ∷ (π₁\;\zs'), \cons\;y\;\zs')
          ==
          \cons\;y\;(x ∷ (π₁\;\zs'), \cons\;x\;\zs')
      }
      \infer1{\mathsf{Bagelim}\,\nil\,\cons\,\resp :
        ∏_{x : \Bag} P\,x}
    \end{prooftree}
  \end{center}
  And it computes:
  \begin{equation*}
    \begin{aligned}
      {}&\mathsf{Bagelim}\,\nil\,\cons\,\resp\,[]
        &&=_{P\,[]} &&\nil
      \\&\mathsf{Bagelim}\,\nil\,\cons\,\resp\,(x ∷ \xs)
        &&=_{P\,(x ∷ \xs)} &&\cons\,x\,(\xs ,
        \mathsf{Bagelim}\,\nil\,\cons\,\resp\,\xs)
    \end{aligned}
  \end{equation*}
\end{example}
\end{samepage}


\subsection{From QIT to QWI-type}

We claim that any QIT \(Q\) in the sense of \Cref{def_equational_QITs}
can be constructed as the QW-type for a signature and equational
system derived from the declaration of the QIT;\@ and that the same
is true for the indexed version of \Cref{def_equational_QITs} and
QWI-types.  That is, QWI-types are universal for non-conditional QITs
in the same sense that W-types are for inductive types in a
sufficiently extensional type theory.

We saw above that the data $c_1:C_1,\ldots,c_m:C_m$ in
\Cref{def_equational_QITs} gives rise to an underlying inductive type
\(⌊Q⌋\) and hence to a W-type~\parencite{altenkirch2015indexed}, with
signature $Σ=(A,B)$.  The parameters and arguments of the equality
constructors $d_1:D_1,\ldots,d_n:D_n$ are also encoded in the same way
by a signature $(E, V)$. Then the endpoints of the equality
constructors can be encoded by the \(l\) and \(r\) arguments of an
equational system $ε=(E,V,l,r)$ in the sense of \Cref{def:syse}; the
encoding follows the structure of $\EqType$ judgements in
\Cref{qit-constructor-types}, using the $η$ constructor of
$T_Σ$~\eqref{eq:Tconstructors} for variables and the $σ$ constructor
for \(c₁, ⋯, cₘ\) introduced by \EqCon{} in \(Θ\). We illustrate the
encoding by example, beginning with the three examples from the
Introduction.

\begin{example}[\textbf{Finite multisets}]
  The element constructors of the QIT, $\Bag\,X$, of finite multisets
  over \(X : 𝓤\) in~\eqref{bag_qit} are encoded exactly as the W-type
  for List over \(X\): we take \(A:𝓤\) to be \(𝟙 + X\), where
  \(ι₁\,0\) corresponds to \([]\) and \(ι₂\,x\) corresponds to
  \(x ∷ \_\) for each \(x : X\). The arity of \([]\) is zero, and the
  arity of each \(x ∷ \_\) is one; so we take \(B : A → 𝓤\) to be the
  function mapping \(ι₁\,0\) to $𝟘$ and each \(ι₂\,x\) to~$𝟙$. The
  \swap{} equality constructor is parametrised by elements of
  \(E ≝ X × X\) and for each \((x, y) : E\), \(\swap\,(x,y)\) yields
  an equation involving a single free variable (called
  \(\zs : \Bag\,X\) in~\eqref{bag_qit}); so we define
  \(V ≝ λ\,\_·𝟙 : E → 𝓤\).  Each side of the equation named by
  \(\swap\,(x,y)\) is coded by an element of
  \(T_Σ\,(V\,(x,y)) = T_Σ\,𝟙\). Recalling the definition of \(T_Σ\)
  from~\cref{sec:w-types}, the single free variable, \(\zs\),
  corresponds to \(η\,0 : T_Σ\,𝟙\). Then the left-hand side of the
  equation, \(x ∷ y ∷ zs\), is encoded as
  \(σ\,(ι₂\,x, (λ\,\_·σ\,(ι₂\,y, (λ\,\_·η\,0))))\), and similarly the
  right-hand side, \(y ∷ x ∷ zs\), is encoded as
  \(σ\,(ι₂\,y, (λ\,\_·σ\,(ι₂\,x, (λ\,\_·η\,0))))\).

  So altogether, the encoding of \(\Bag\,X\) as a QW-type uses the
  non-indexed signature $Σ=(A,B)$ and equational system $ε=(E,V,l,r)$,
  where:
  \begin{align*}
    A &≝ 𝟙 + X & E &≝ X × X\\
    B (ι₁\,0) &≝ 𝟘 & V(x, y) &≝ 𝟙\\
    B(ι₂\,x)  &≝ 𝟙 &l (x, y)
    &≝ σ(ι₂\,x , (λ\,\_·\,σ(ι₂\,y , (λ\,\_·\,η\,0))))\\
    &&r (x, y) &≝ σ(ι₂\,y , (λ\,\_·\,σ(ι₂\,x , (λ\,\_·\,η\,0))))
  \end{align*}
\end{example}

\begin{example}[\textbf{Length-indexed multisets}]%
  \label{exa:length-indexed-multisets}
  The QWI-type
  encoding the QIT $\CommVec\,X$ of length-indexed multisets in~\eqref{abvec_qit} is an indexed version of the previous example,
  using the index type $I ≝ ℕ$. The indexed signature $Σ=(ℕ,A,B)$ has
  $A:𝓤^ℕ$ and $B:∏_{i:ℕ}(A_i→𝓤^ℕ)$ given by:
  \begin{align*}
    A_0   &≝ 𝟙\\
    A_{i+1} &≝ X\\
    B_0\,0\,j &≝ 𝟘\\
    B_{i+1}\,x\,j &≝ (i = j)
  \end{align*}
  The indexed system of equations $ε=(ℕ,E,V,l,r)$ over $Σ$ has $E:𝓤^ℕ$,
  $V:∏_{i:ℕ}(Eᵢ → 𝓤^ℕ)$ and $l,r:∏_{i:ℕ}∏_{e:Eᵢ}(T_Σ(Vᵢ\,e))ᵢ$ given
  by:
  \begin{align*}
    E_0 &≝ 𝟘\\
    E_1 &≝ 𝟘\\
    E_{i+2} &≝ X × X\\
    V_{i+2}\,(x,y)\,j &≝ (i = j)\\
    l_{i+2}\,(x,y) &≝ σ_{i+2}(x, λ\,\_·λ\,\,\mathsf{refl}·\,σ_{i+1}(y,
                     λ\,\_·λ\,\,\mathsf{refl}·\,η_i\,\mathsf{refl}))\\
    r_{i+2}\,(x,y) &≝ σ_{i+2}(y, λ\,\_·λ\,\mathsf{refl}·\,σ_{i+1}(x,
                     λ\,\_·λ\,\mathsf{refl}·\,η_i\,\mathsf{refl}))
  \end{align*}
  (We have used dependent pattern matching~\parencite{CoquandT:patmdt}
  to simplify the formulation of the above definitions;
  $\mathsf{refl}:x=x$ denotes the proof of reflexivity; $V_0$, $V_1$,
  $l_0$ and $l_1$ are uniquely determined.)
\end{example}

\begin{example}[\textbf{Unordered countably-branching trees}]
  The parameters of the \leaf{} and \node{} constructors for the QIT
  in~\eqref{infTree_qit} look like those for \Bag{}, except that the
  arity of  \node{} is ℕ rather than 𝟙. So we use the non-indexed
  signature $Σ=(A,B)$ with $A:𝓤$ and $B:A→ 𝓤$ given by:
  \begin{align*}
    A &≝ 𝟙 + X\\
    B (ι₁\,0) &≝ 𝟘\\
    B (ι₂\,x)  &≝ ℕ
  \end{align*}
  The $\perm$ equality constructor is parameterised by elements of
  $X × ∑_{b : ℕ → ℕ} \isIso\,b$.  For each element $(x,b,b')$ of that
  type, $\perm(x,b,b')$ yields an equation involving an
  $ℕ$-indexed family of variables (called $f:ℕ→\infTree\,X$ in~\eqref{infTree_qit}); so we take $V:E→𝓤$ to be $λ\,\_·\,ℕ$. Each side
  of the equation named by $\perm(x,b,b')$ is coded by an element of
  $T_Σ(V(x,b,b')) = T_Σ(ℕ)$.  The $ℕ$-indexed family of variables is
  represented by the function $η: ℕ→T_Σ(ℕ)$ and its permuted version
  by $η∘b$. Thus the left-{} and right-hand sides of the equation
  named by $\perm(x,b,b')$ are coded respectively by the elements
  $σ(ι₂\,x,η)$ and $σ(ι₂\,x,η∘b)$. Finally, the non-indexed system of
  equations over $Σ$ for the QW-type corresponding to the QIT in~\eqref{infTree_qit} is $ε=(E,V,l,r)$ with:
  \begin{align*}
    E &≝ X × ∑_{b : ℕ → ℕ} \isIso\,b \\
    V (x,b,b') &≝ ℕ\\
    l (x,b,b') &≝ σ(ι₂\,x,η)\\
    r (x,b,b') &≝ σ(ι₂\,x,η∘b)
  \end{align*}
\end{example}
This example is significant since, prior to the introduction of
QW-types~\parencite{fiore2020constructing}, none of the constructions
in the existing literature on subclasses of
QITs~\parencite{sojakova2015higher,swan2018wtypes,dybjer2018finitary}
or QIITs~\parencite{dijkstra2017quotient,kaposi2019constructing}
supported infinitary QITs.

\begin{example}[\textbf{W-suspensions}~\parencite{sojakova2015higher}]%
    \label{exa:wsusp}
    W-suspensions are also instances of QW-types. The data for one of
    them is: a type \(A' : 𝓤\) and a type family \(B' : A' → 𝓤\),
    (that is, the same data as a W-type), together with a type
    \(C' : 𝓤\) and functions \(l', r' : C' → A'\) that express a
    restricted subset of equalities. The equivalent QW-type is:
    \begin{align*}
        {} A &≝ A' & E &≝ C'
            & l &≝ λ\,c · σ\,(l'\,c, η)
        \\ B &≝ B' & V &≝ λ\,c·B'\,(l'\,c) × B'\,(r'\,c)
            & r &≝ λ\,c · σ\,(r'\,c, η)
    \end{align*}
\end{example}

\begin{example}[\textbf{W-types with reductions}~\parencite{swan2018wtypes}]
  W-types with reductions are further instances of QWI-types. The data
  for such a type, using our notational conventions, is:
  \(Z : 𝓤, Y : 𝓤^Z, X : Y ⥗ 𝓤^Z\), together with, for all \(z : Z\) a
  reindexing map \(R_z : ∏_{y : Y_z} (X_{z}\,y)_z\). A reduction
  identifies a term \(σ_z\,(y,α)\) for the signature $(Z,Y,X)$ at
  index \(z\) with the argument that was used to construct that term
  at the re-indexing map, namely \(α_z\,(R_z\,y)\). We claim that the
  equivalent QWI-type is given by:
    \begin{align*}
        {} I &≝ Z & A &≝ Y & E &≝ Y & l &≝ λ\,z · λ\,y · σ_z\,(y, η)
        \\ && B &≝ X & V &≝ X & r &≝ λ\,z · λ\,y · η_z\,(R_z\,y)
    \end{align*}
\end{example}



\begin{example}[\textbf{Blass, Lumsdaine and Shulman's HIT}~\parencite{lumsdaine2019semantics}]%
  \label{blass_lumsdaine_shulman_type}

  As discussed in \cref{sec:wisc}, \textcite[Section~9]{BlassA:worfac}
  shows that, provided a certain large cardinal axiom is consistent
  with ZFC, then there is an infinitary equational theory with no
  initial algebra in ZF\@.  \textcite[Section~9]{lumsdaine2019semantics}
  adapt this theory into a higher-inductive type, called \(F\),
  that cannot be proved to exist in ZF, and hence cannot be
  constructed in type theory just using pushouts and natural
  numbers. Nevertheless it can be constructed in the type theory of
  \Cref{sec:type-theory} plus WISC, because of \Cref{thm:main} and the
  fact that \(F\) can be expressed as a QWI-type.

  The type \(F\) can be thought of as a set of notations for countable
  ordinals. Its definition requires some setup.

  Fix the bijection \([\even,\odd] : ℕ + ℕ → ℕ\) with inverse \(\eoinv\). Fix
  another bijection \(\pair : ℕ × ℕ → ℕ\) with inverse \(⟨\unpaira,\unpairb⟩\).
  Given \(f, g : ℕ → A\), write \(f∪g : ℕ → A\) for the composite
  \([f,g]∘\eoinv\); given \(a : A\), write \(\{a\} : ℕ → A\) for the constant
  function at \(a\).

  Now \(F\) is defined by the three point constructors,
  \begin{equation*}
    \begin{aligned}
      {} 0 &: F
      \\ S &: F → F
      \\ \fsup &: (ℕ → F) → F
    \end{aligned}
  \end{equation*}
  and the five path constructors,
  \begin{equation*}
    \begin{aligned}
      {}& \fsup\,\{0\} = 0
      \\
      {}& ∏_{f,g : ℕ → ℕ}
        \mathsf{E}\,f\,g
      → ∏_{h : ℕ → F} \fsup\,(h∘f) = \fsup\,(h∘g)
      \\& ∏_{f,g : ℕ → F} \fsup\,(f ∪ \{\fsup\,(f ∪ g)\})
      = \fsup\,(f ∪ g)
      \\& ∏_{f,g : ℕ → F} \fsup\,(f ∪ \{S\,(\fsup\,(f ∪ g))\})
      = S\,(\fsup\,(f ∪ g))
      \\& ∏_{b,c : ℕ → ℕ} ∏_{L : ℕ → ℕ → ℕ}
      \jointsurj\,b\,c → \Lrel\,L\,b\,c → \Lrel\,L\,c\,b
      → ∏_{h : ℕ → F} \fsup\,\left(\iter\,h\,b\right)
      = \fsup\,\left(\iter\,h\,c\right)
    \end{aligned}
  \end{equation*}
  where
  \begin{equation*}
    \begin{aligned}
      \mathsf{E}\,f\,g &≝ ∏_{n : ℕ}\left( ∑_{m : ℕ} f(m) = n \right)
        ↔ \left( ∑_{m : ℕ} g(m) = n \right)\\
        {}  \jointsurj\,b\,c &≝{}
        ∏_{n : ℕ} \left( ∑_{m : ℕ} b(m) = n \right) ∨ \left( ∑_{m : ℕ} c(m) = n \right)
      \\  \Lrel\,L\,x\,y &≝ ∏_{n : ℕ} ∑_{m,l : ℕ}
      \left( L\,(x\,m)\,l = y\,n \right)
      \\  \iter\,h\,p\, &≝ ƛ n · h_{(\unpaira\,n)}\,(p\,(\unpairb\,n))
      \\  h₀ &≝ h
      \\  h_{(x + 1)} &≝ ƛ y · \fsup\,(hₓ ∘ (L\,y))
    \end{aligned}
  \end{equation*}
  This can be expressed as a QW-type as follows:
  %
    \begin{align*}
      {}  A &≝ 𝟛
      \\  B &≝ [𝟘,𝟙,ℕ]
      \\  E &≝ 𝟙 + \left(∑_{f,g : ℕ → ℕ}\mathsf{E}\,f\,g\right)
              + 𝟙 + 𝟙 + \left(∑_{b,c : ℕ → ℕ} ∑_{L : ℕ → ℕ → ℕ}
              \left(\jointsurj\,b\,c × \Lrel\,L\,b\,c × \Lrel\,L\,c\,b
              \right)\right)
      \\  V &≝ [𝟘, ℕ, ℕ + ℕ, ℕ + ℕ, ℕ]
    \end{align*}
    \begin{equation*}
      \begin{array}{@{}l@{\,}l@{}r@{}l@{}}
        {}  l & (ι₁\,\_) &{}≝{}
        &   σ\,(2,ƛ\_·σ\,(0,!))
        \\  r & (ι₁\,\_) &{}≝{}
        &   σ\,(0,!)
        \\[\jot]  l & (ι₂\,f\,g) &{}≝{}
        &   σ\,(2,(η∘f))
        \\  r & (ι₂\,f\,g) &{}≝{}
        &   σ\,(2,(η∘g))
        \\[\jot]  l & (ι₃\,\_) &{}≝{}
        &   σ\,(2,((η ∘ ι₁) ∪ ƛ\_· σ\,(1,((η ∘ ι₁) ∪ (η ∘ ι₂)))))
        \\  r & (ι₃\,\_) &{}≝{}
        &   σ\,(2,((η ∘ ι₁) ∪ (η ∘ ι₂)))
        \\[\jot]  l & (ι₄\,\_) &{}≝{}
        &   σ\,(2,((η ∘ ι₁) ∪ ƛ\_·σ\,(1,σ\,(2,((η ∘ ι₁) ∪ (η ∘ ι₂))))))
        \\  r & (ι₄\,\_) &{}≝{}
        &   σ\,(1,σ\,(2,((η ∘ ι₁) ∪ (η ∘ ι₂))))
        \\[\jot]  l & (ι₅\,b\,c\,L\,\_\,\_\,\_) &{}≝{}
        &   σ\,(2,ƛn· k\,L\,(\unpaira\,n)\,(b\,(\unpairb\,n)))
        \\  r & (ι₅\,b\,c\,L\,\_\,\_\,\_) &{}≝{}
        &   σ\,(2,ƛn· k\,L\,(\unpaira\,n)\,(c\,(\unpairb\,n)))
      \end{array}
    \end{equation*}
    where
    \begin{equation*}
      \begin{array}{@{}l@{}c@{}l@{}}
        {}  k\,L\,0\, &{}≝{}& η
        \\  k\,L\,(x + 1) &{}≝{}& ƛ y · σ\,(2,((k\,L\,x) ∘ (L\,y)))
      \end{array}
    \end{equation*}
\end{example}



\FloatBarrier
\section{Conclusion}%
\label{sec:conclusion}

QWI-types are a general form of indexed quotient inductive types that
capture many examples, including simple 1-cell complexes and
non-recursive QITs~\parencite{basold2017higher}, non-structural
QITs~\parencite{sojakova2015higher}, W-types with
reductions~\parencite{swan2018wtypes} and also infinitary QITs
(e.g.~unordered infinitely branching
trees~\parencite{altenkirch2016type} and the type $F$
of \textcite{BlassA:worfac}, \textcite{lumsdaine2019semantics}).  We
have shown that it is possible to construct any QWI-type, even
infinitary ones, in the extensional type theory of toposes with
natural numbers object and universes satisfying the WISC Axiom.
(Indeed our Agda development shows that the construction is also
possible in weaker, more intensional type theories.)

We conclude by
mentioning related work and some possible directions for future work.

\subsection*{Reduction of QWI to QW}
WI-types (indexed W-types) can be constructed from \mbox{W-types}:
see~\parencite[Theorem~12]{GambinoN:welftd} and
\parencite[Section~7.1]{altenkirch2015indexed}. Our main
\Cref{thm:main} implies indirectly that in the presence of WISC,
QWI-types can be constructed from QW-types (since quotients and
W-types are instances of the latter). We do not know whether there is
a direct reduction of QWI-types to QW-types without any further
axioms, extending the reduction of WI-types to W-types.

\subsection*{Conditional path equations}

In \Cref{sec:general-qit-schemas} we mentioned the fact that
\textcite{dybjer2018finitary} give a schema for finitary 1-HITs (and
2-HITs) in which path constructors are allowed to take arguments
involving the identity type of the datatype being declared. To cover
this case, one could consider generalising QWI-types by replacing the
role played by infinitary equational theories by theories whose axioms
are infinitary Horn clauses, that is, equations conditioned by
(possibly infinite) conjunctions of equations. In fact, we have
already extended \Cref{thm:main} to cover this case, and this will
likely appear in the third author's upcoming thesis.

\subsection*{Quotients of monads}

\Cref{thm:main} gives a construction of initial algebras for
equational systems on the \emph{free} monad $T_Σ$ generated by a
signature $Σ$. By a suitable change of signature (see
\Cref{remarkFreeAlgebras}) this extends to a construction of free
algebras, rather than just initial ones.  One can show that the
construction works not just for free monads, but for more general ones
where the underlying endofunctor is well-behaved (satisfies some
suitable colimit preservation), such as the quotient monads one
obtains from systems of equations over signatures.  Given such a
construction one gets a quotient monad morphism from the base monad to
the quotient monad. This contravariantly induces a forgetful functor
from the algebras of the latter to those of the former.  Using the
adjoint triangle theorem, one should be able to construct a left
adjoint. This would then cover examples such as the free group over a
monoid, free ring over a group, etc.

\subsection*{Quotient inductive-inductive types}

We do not know whether our analysis of QITs using quotients, inductive
types and a well-founded notion of size can be extended to cover the
notion of \emph{quotient inductive-inductive} type
(QIIT)~\parencite{AltenkirchT:quoiit,kaposi2019constructing}.
\textcite{dijkstra2017quotient} studies such types in depth and in
Chapter~6 of his thesis gives a construction for finitary ones in
terms of countable colimits, and hence in terms of countable
coproducts and quotients. One could hope to pass to the infinitary
case by using well-founded sizes as we have done, provided an analogue
for QIITs can be found of the construction in \Cref{sec:construction}
for our class of QITs, the QWI-types.
\Textcite{kaposi2019constructing,KaposiA:lariqi} give a specification
of QIITs using a domain-specific type theory called the \emph{theory
of signatures} (ToS) and, assuming the ToS exists, prove existence of
all QIITs matching this specification. It might be possible to encode
their ToS using QWI-types (it is already an example of a QIIT), or to
extend QWI-types to a notion of ``inductive-inductive QW-type'' that
makes this possible. Together this would provide a construction of all
QIITs.

\subsection*{Homotopy Type Theory (HoTT)}

In this paper we have used an extensional type theory. Our Agda
development is more intensional, but nevertheless makes use of the
Uniqueness of Identity Proofs (UIP) axiom, which is well-known to be
incompatible with the Univalence
Axiom~\parencite[Example~3.1.9]{HoTT}. Given the interest in HoTT, it
is certainly worth investigating whether an analogue of
\Cref{thm:main} holds in some model of univalent foundations.

\subsection*{Pattern matching for QITs and HITs}

Our reduction of QITs to quotients and inductive types in the presence
of WISC is of foundational interest. For applications, one could wish
for direct support in systems like Agda, Coq and Lean for the very
useful notion of quotient inductive type, or more generally, for
higher inductive types. Even having better support for the special
case of quotient types would be welcome.  It is not hard to envisage
the addition of a general schema for declaring QITs; but when it comes
to defining functions on them, having to do that with eliminator forms
rapidly becomes cumbersome (for example, for functions of several QIT
arguments). Some extension of dependently-typed pattern matching to
cover equality constructors as well as element constructors is
needed. In this context it is worth mentioning that the
\texttt{cubical} features of recent versions of Agda give access to
cubical type theory~\parencite{MortbergA:cubadt}. This allows for easy
declaration of HITs (and hence in particular QITs) and a certain amount
of pattern matching when it comes to defining functions on them: the
value of a function on a path constructor can be specified by using
generic elements of the interval type in point-level patterns; but
currently the user is given little mechanised assistance to solve the
definitional equality constraints on end-points of paths that are
generated by this method.


\FloatBarrier
\section*{Acknowledgement}

We would like to acknowledge the contribution Ian Orton made to the
initial development of the work described here; he and the first
author supervised the third author's Master's thesis in which the
notion of a QW-type was first introduced. The second author would also
like to thank Paul Blain Levy for discussions about the various forms
of the \emph{weakly initial sets of covers} axiom
(see~\Cref{rem:global-wisc}); and Andrew Swan for discussions about
the role of that axiom in constructive proofs of cocontinuity for
polynominal endofunctors (see~\Cref{cor:poly-cocontinuity}).

\FloatBarrier
\begin{emergency}{2em} 
\printbibliography%
\end{emergency}

\end{document}